\documentclass{article}
\usepackage[letterpaper,margin=1in]{geometry}
\usepackage{graphicx} % Required for inserting images
\usepackage[numbers,sort,compress]{natbib}
\usepackage{pslatex}
\usepackage{amsfonts,color,morefloats}
\usepackage{amssymb,amsmath,latexsym,amsthm}
\usepackage{stmaryrd}
\usepackage{cases}
\usepackage{hyperref}
\usepackage{color}
\usepackage{makecell}
\usepackage{bm}
\usepackage{comment}
\usepackage{enumitem}

\newtheorem{theorem}{Theorem}
\newtheorem{lemma}{Lemma}
\newtheorem{proposition}{Proposition}
\newtheorem{corollary}{Corollary}

\newtheorem{definition}{Definition}

\newtheorem{result}{Result}
\newtheorem{remark}{Remark}

\newcommand{\F}{\mathbb{{F}}}
\newcommand{\C}{\mathbb{{C}}}
\newcommand{\Q}{\mathbb{{Q}}}
\newcommand{\Fp}{\F_p}
\newcommand{\Fq}{\F_q}
\newcommand{\FQ}{\F_Q}

\newcommand{\Ft}{\F_2}
\newcommand{\Fqm}{\F_{q^m}}

\newcommand{\al}{{\alpha}}
\newcommand{\be}{{\beta}}
\newcommand{\ga}{{\gamma}}
\newcommand{\la}{{\lambda}}
\newcommand{\om}{{\omega}}
\newcommand{\sig}{{\sigma}}

\newcommand{\spa}{\text{span}}
\newcommand{\supp}{\text{Supp}}
\newcommand{\zp}{{\zeta_p}}

\newcommand{\lcm}{{\mathrm{lcm}}}

\newcommand{\cC}{{\mathcal{C}}}

\newcommand{\cR}{{\mathcal{R}}}
\newcommand{\cS}{{\mathcal{S}}}

\newcommand{\es}{\emptyset}
\newcommand{\lc}{\lceil}
\newcommand{\ol}{\overline}
\newcommand{\rc}{\rceil}

\newcommand{\sm}{\setminus}
\newcommand{\td}{\tilde}
\newcommand{\Tr}{{\mathrm{Tr}}}
\newcommand{\Trqp}{{\Tr_{q/p}}}
\newcommand{\Trqt}{{\Tr_{q/2}}}
\newcommand{\TrQt}{{\Tr_{Q/2}}}

\newcommand{\ba}{{\bf a}}
\newcommand{\bc}{{\bf c}}
\newcommand{\bh}{{\bf h}}
\newcommand{\bs}{{\bf s}}

\newcommand{\bv}{{\bf v}}
\newcommand{\bw}{{\bf w}}
\newcommand{\bx}{{\bf x}}
\newcommand{\by}{{\bf y}}
\newcommand{\bz}{{\bf 0}}

\newcommand{\wt}{\widetilde}
\newcommand{\wth}{\text{wt}_H}

\newcommand{\ud}{\mathrm{d}}

\newcommand{\qbinom}[2]{%
  \left[ \genfrac{}{}{0pt}{}{#1}{#2} \right]%
}

\title{The generalized covering radii of Melas codes}
\author{
Shuxing Li\thanks{Department of Mathematical Sciences, University of Delaware, Newark, DE 19716, USA (email: shuxingl@udel.edu). The work of Shuxing Li was supported by U.S. National Science
Foundation under Grant DMS-2452236 and the University of Delaware Research Foundation Strategic Initiative (UDRF-SI) program.}
\and
Maosheng Xiong\thanks{Department of Mathematics, The Hong Kong University of Science and Technology, Hong Kong (email: mamsxiong@ust.hk). The work of Maosheng Xiong was supported by the Research Grants Council (RGC) of Hong Kong under Grant 16306626.}
}

\date{}

\begin{document}

\maketitle

\begin{abstract}
The generalized covering radii have recently emerged as fundamental parameters of linear codes with applications to database linear querying. In this paper, we study the generalized covering radii $\rho_t(M(m,q))$ of Melas codes $M(m,q)$ over any finite field $\mathbb{F}_q$. We determine $\rho_2(M(m,q))$ for all $q$, and for a general $t \ge 3$, we prove that $\rho_t(M(m,q)) \in \left\{2t,2t+1\right\}$ for $q \in \{2,3\}$ and $\rho_t(M(m,q))=2t$ for $q \ge 4$ whenever $m$ is sufficiently large. These results extend recent work on the covering radius of Melas codes.  
\end{abstract}

\section{Introduction} \label{sec-intro}
The notion of generalized covering radii was recently introduced by Elimelech, Firer, and Schwartz as a fundamental property of linear codes \cite{EFS21}. It was originally motivated by applications to database linear querying, including private information-retrieval protocols, where the goal is to reduce access complexity. Since then, generalized covering radii have attracted considerable attention. In particular, extensive work has examined the generalized covering radii of specific families of codes, including Reed-Muller codes \cite{EWS22}, binary primitive double-error-correcting BCH codes \cite{OO26,XY26,YS25}, binary primitive triple-error-correcting BCH codes \cite{EZ26,OO26b}, binary primitive BCH codes with higher orders \cite{BZ26,XYZ26}, certain binary cyclic codes \cite{LZM+26}, binary Melas codes \cite{LYYCO26}, and generalized Zetterberg codes \cite{LLLS26}. A geometric approach to generalized covering radii of linear codes has been presented in \cite{AMNT26}. Moreover, the optimal asymptotic rate of generalized covering codes of every order, for both unrestricted and linear codes, is determined in \cite{LSW26}. For recent advances on the connection between generalized covering radii and generalized packing radii, please see \cite{EZ26b,YS26}.

Let $\Fq$ be the finite field of $q$ elements. For a positive integer $n$, we use $[n]$ to denote the set of consecutive positive integers $\{1,2,\ldots,n\}$. Recall that the Hamming weight of $\bx=(x_1,x_2,\ldots,x_n) \in \Fq^n$ is defined to be
$$
\wth(\bx) \triangleq |\{ i \in [n] \mid x_i \ne 0 \}|
$$
and the Hamming distance between $\bx$ and $\by$ in $\Fq^n$ is defined to be
$$
\ud_H(\bx,\by) \triangleq \wth(\bx-\by).
$$
An $[n,k,d]_q$ linear code $\cC$ is a $k$-dimensional vector subspace of the ambient space $\Fq^n$ over $\Fq$, such that its minimum distance $d$ is given by 
$$
d=\min_{\substack{ \bx, \by \in \cC \\ \bx \ne \by}} \ud_H(\bx,\by).
$$
When the minimum distance is not specified, we refer to $\mathcal{C}$ simply as an $[n,k]_q$ linear code. While the minimum distance describes the distance among the vectors within the linear code $\cC$, the \emph{covering radius} of the linear code $\cC$ reflects the distance between the code $\cC$ and the ambient space $\Fq^n$. Indeed, the covering radius $\rho_1(\cC)$ of the linear code $\cC$ is defined as
$$
\rho_1(\cC) \triangleq \max_{\bx \in \Fq^n} \min_{\bc \in \cC} \ud_H(\bx,\bc).
$$
Equivalently, the covering radius of the linear code $\cC$ can be defined in the following way.

\begin{definition}[covering radius]
\label{def-cr}
Let $\cC$ be an $[n,k]_q$ linear code with a parity-check matrix $H$ having columns $\bh_1,\ldots,\bh_n \in \Fq^{n-k}$. The covering radius $\rho_1(\cC)$ of $\cC$ is the smallest nonnegative integer $r$ such that for each $\bs \in \Fq^{n-k}$, there exists a subset $I  \subseteq [n]$ with $|I| \le r$ satisfying $\bs \in \langle \bh_{i} \mid i \in I \rangle$.
\end{definition}

Similar to the minimum distance, the covering radius is a fundamental metric of linear codes that measures how well the codewords are spread throughout the ambient space. There has been a significant body of research regarding the covering radius, see for instance \cite{BLP98,CHLL97,CKMS85,CLLM97} and the references therein.

Motivated by the study of database linear querying, the notion of covering radius of linear code has been generalized recently in \cite{EFS21} by Elimelech, Firer, and Schwartz. 

\begin{definition}[generalized covering radius of linear codes {\cite[Definition 1]{EFS21}}]
\label{def-gcr1}
Let $\cC$ be an $[n,k]_q$ linear code with a parity-check matrix $H$ having $n$ columns $\bh_1, \bh_2, \ldots, \bh_n \in \Fq^{n-k}$. The $t$-th  generalized covering radius $\rho_t(\cC)$ of $\cC$, is the smallest nonnegative integer $r$ such that for any $t$ vectors $\bs_1, \bs_2, \ldots, \bs_t \in \Fq^{n-k}$, there exists a subset $I  \subseteq [n]$ with $|I| \le r$ satisfying $\{\bs_1,\ldots,\bs_t\} \subseteq \langle \bh_{i} \mid i \in I \rangle$. 
\end{definition}   

Consider a database whose queries are linear combinations of the columns of a matrix $H$. Regard $H$ as the parity-check matrix of a linear code $\cC$. The $t$-th generalized covering radius $\rho_t(\cC) = r$ has a natural interpretation in this setting: given any batch of $t$ queries from users, it suffices to access at most $r$ columns of $H$ in order to answer all $t$ queries. Thus, the generalized covering radius serves as a measure of access complexity.

\begin{remark}\label{rem-gcr}
Let $\cC$ be an $[n,k]_q$ linear code.
\begin{enumerate}[label=(\arabic*)]
\item It is remarked in \cite[Lemma 2]{EFS21} that the generalized covering radius $\rho_t(\cC)$ does not depend on the specific choice of the parity-check matrix $H$ in Definition \ref{def-gcr1}. 
\item By Definitions \ref{def-cr}, \ref{def-gcr1}, the first generalized covering radius $\rho_1(\cC)$ is precisely the covering radius of the linear code $\cC$. Moreover, the following monotonicity holds:
$$
0 \le \rho_1(\cC) \le \rho_2(\cC) \le \cdots \le \rho_{n-k}(\cC)=n-k.
$$
Note that $\rho_t(\cC)=0$ for some $1 \le t \le n-k$ if and only if $\cC$ is the ambient space $\Fq^n$. For simplicity, we set $\rho_t(\cC)=n-k$ for $t \ge n-k$. 
\item By Definition \ref{def-gcr1} and \cite[Proposition 15]{EFS21}, given an $[n,k]_q$ linear code $\cC$, for $t_1, t_2 \ge 1$ with $t_1+t_2 \le n-k$, we have 
\begin{equation}
\label{eqn-additive}
\rho_{t_1+t_2}(\cC) \le \rho_{t_1}(\cC)+\rho_{t_2}(\cC).
\end{equation}
In particular, for $1 \le t \le n-k$ and $s \mid t$, we have
$$
\rho_t(\cC) \le \min \left\{ \frac{t}{s} \rho_s(\cC), n-k \right\}.
$$
\end{enumerate}
\end{remark}

In this paper, we focus on the well-known family of Melas codes \cite{Melas60}. Let $M(m,q)$ denote the $q$-ary Melas code of length $q^m-1$. The covering radii of Melas code $M(m,q)$ have been determined for all possible $(m,q)$ pairs, see \cite{SHOS22} and the references therein. We are able to completely determine the second generalized covering radii of Melas code $M(m,q)$ for all possible $(m,q)$ pairs.

\begin{theorem}
\label{thm-seccr}
Let $m \ge 1$ and $q$ be a prime power. Then the following holds
\begin{equation}
\label{eqn-seccr}
\rho_2(M(m,q))=
\begin{cases}
                     1 & \mbox{if $q=2$, $m=1$} \\
                     2 & \mbox{if $q=2$, $m=2$} \\
                     5 & \mbox{if $q=2$, $m=3$} \\
                     6 & \mbox{if $q=2$, $m=4$} \\ 
                     5 & \mbox{if $q=2$, $m \ge 5$} \\
                     1 & \mbox{if $q=3$, $m=1$} \\
                     4  & \mbox{if $q=3$, $m=2$} \\
                     5  & \mbox{if $q=3$, $m \ge 3$} \\
                     2 & \mbox{if $q \ge 4$ and $m = 1$}\\
                     4 & \mbox{if $q \ge 4$ and $m \ge 2$}
\end{cases}
\end{equation}
\end{theorem}

Theorem \ref{thm-seccr} is derived from a general approach yielding almost matching lower and upper bounds on the $t$-th generalized covering radius of $M(m,q)$ when $m$ is sufficiently large compared with $t$. Specifically, a combinatorial technique first proposed in \cite[Theorem IV.1]{XY26} leads to the following lower bound.

\begin{theorem}
\label{thm-asymlb}
Let $q$ be a prime power and $m \ge 1$. For $1 \le t \le m$ and $(m,q) \notin \{(2,2), (1,2), (1,3)\}$, if 
$$	
q^m \ge \frac{q^{t(2t-1)}}{(2t-1)!},
$$
then $\rho_t(M(m,q)) \geq 2t$. In particular, if $m \ge t(2t-1)$ and $(m,q) \notin \{(2,2), (1,2), (1,3)\}$, then $\rho_t(M(m,q)) \geq 2t$.
\end{theorem}

On the other hand, a careful character sum analysis establishes the following upper bound.

\begin{theorem}
\label{thm-asymub}
Let $q$ be a prime power and $m \ge 1$. For $1 \le t \le m$, if 
$$	
m \ge 2(t+1) \log_{q}2+2\log_qt,
$$
then $\rho_t(M(m,q)) \le 2t+1$. 
\end{theorem}

Theorems \ref{thm-asymlb}, \ref{thm-asymub} can be applied to establish the following range of $\rho_t(M(m,q))$ with $t \ge 3$,  tailored for the three cases: $q=2$, $q=3$, or $q \ge 4$. We use $[a,b]$ to denote the interval that includes integers $a$ and $b$. 

\begin{theorem}
\label{thm-asymcr}
\begin{enumerate}[label=(\arabic*)]
\item Let $m \ge 3$ and $3 \le t \le m$. Then 
$$
\rho_t(M(m,2)) \in 
\begin{cases}
\left[t+\lc \log_2 (t+1) \rc,\min\{3t,2m\}\right] & \mbox{if $3 \le m < 2(t+1)+2\log_2 t$ } \\
\left[t+\lc \log_2 (t+1) \rc,2t+1\right] & \mbox{if $2(t+1)+2\log_2 t \le m < t(2t-1)$} \\
\left[2t,2t+1\right] & \mbox{if $m \ge t(2t-1)$}
\end{cases}
$$
\item Let $m \ge 3$ and $3 \le t \le m$. Then 
$$
\rho_t(M(m,3)) \in 
\begin{cases}
\left[t+\lc \log_3 (2t+3) \rc,\min\{3t,2m\}\right] & \mbox{if $3 \le m < 2(t+1)\log_3 2+2\log_3 t$ } \\
\left[t+\lc \log_3 (2t+3) \rc,2t+1\right] & \mbox{if $2(t+1)\log_3 2+2\log_3 t \le m < t(2t-1)$} \\
\left[2t,2t+1\right] & \mbox{if $m \ge t(2t-1)$}
\end{cases}
$$
\item Let $q \ge 4$ be a prime power. Let $m \ge 3$ and $3 \le t \le m$. Then 
$$
\rho_t(M(m,q))  
\begin{cases}
\in \left[t+\lc \log_q (t+1) \rc,2t\right] & \mbox{if $3 \le m < t(2t-1)$ } \\
= 2t & \mbox{if $m \ge t(2t-1)$}
\end{cases}
$$
\end{enumerate}
\end{theorem}

Theorem \ref{thm-asymcr} indicates that as long as $m \ge t(2t-1)$, we have either $\rho_t(M(m,q)) \in [2t,2t+1]$ for $q \in \{2,3\}$, or $\rho_t(M(m,q)) =2t$ for $q \ge 4$. We note that concerning binary Melas codea, some of our results have been reported in prior literatures. For $m \ge 8$, the second generalized covering radius $\rho_2(M(m,2))$ has been determined in \cite[Theorem 16]{LZM+26}. Moreover, the recently appeared paper \cite{LYYCO26} derived $\rho_t(M(m,2)) \in \{2t,2t+1\}$ for $t>2$ whenever $m > 2t^2-t+2$ \cite[Corollary 6]{LYYCO26} and the third generalized covering radius $\rho_3(M(m,2))=7$ whenever $m \ge 4$ \cite[Remark 4]{LYYCO26}. To the best of our knowledge, all results in the current paper concerning $q$-ary Melas code with $q \ge 3$ represent new contributions to the literature.

The remainder of this paper is organized as follows. Section \ref{sec-prelim} reviews the necessary background, including generalized Hamming weights and Melas codes, and presents several auxiliary technical results. Section \ref{sec-bound} establishes crucial lower and upper bounds on the generalized covering radii of Melas codes. Sections \ref{sec-qge4}, \ref{sec-qeq3}, and \ref{sec-qeq2} then study the generalized covering radii of $M(m,q)$ for the cases $q \ge 4$, $q=3$, and $q=2$, respectively. Finally, Section \ref{sec-conclusion} concludes the paper.

\section{Preliminaries} \label{sec-prelim}
\subsection{Some auxiliary results}

Let $q$ be a power of an odd prime $p$. A \emph{multiplicative quadratic character} of $\Fq$ is a group homomorphism $\eta: \Fq^* \mapsto \{\pm1\} \subset \C^*$ satisfying
$$
\eta(x)=\begin{cases}
            1 & \mbox{$x$ is a nonzero square in $\Fq$} \\
            -1 & \mbox{$x$ is a nonsquare in $\Fq$}
        \end{cases} 
$$
Conventionally, the domain of $\eta$ can be extended to $\F_q$ by defining $\eta(0)=0 \in \C$. 

\begin{lemma}[Weil's bound for quadratic characters {\cite[Theorem 5.41]{LN97}}]
\label{lem-Weilquad}
Let $\eta$ be a quadratic character of $\Fq$ with $q$ odd, and let $f \in \Fq[x]$ be a polynomial that is not a square of a polynomial in $\overline{\Fq}[x]$, where $\overline{\Fq}$ denotes the algebraic closure of $\Fq$. Let $d$ be the number of distinct roots of $f$ in $\overline{\Fq}$. Then
$$
\big|\sum_{x\in\Fq}\eta(f(x))\big|\le(d-1)\sqrt{q}.
$$
\end{lemma}

Let $\zp:=\exp\left(2 \pi \sqrt{-1}/p\right)$ be the complex primitive $p$-th root of unity. A \emph{canonical additive character} of $\Fq$ is a group homomorphism $\varphi: \Fq \mapsto \C^*$ defined by
$$
\varphi(x)=\zeta_p^{\Trqp(x)},  \forall x \in \Fq,
$$
where $\Trqp$ is the trace function from $\Fq$ to $\Fp$. The following Weil-type bound for rational functions follows from \cite{CP06}. For a monic irreducible polynomial $u \in \FQ[X]$, let $v_u$ denote the $u$-adic valuation on $\FQ(X)$, so that $v_u(g)$ is the exponent of $u$ in the factorization of $g \in \FQ(X)$ into a fraction of irreducible polynomials. The rational function $g \in \FQ(X)$ has a pole of order $r$ at $u$ when $v_u(g)=-r<0$, which is to say that $g$ has a pole of order $r$ at every root of $u$ in $\ol{\FQ}$.

\begin{lemma}[Weil's bound for rational functions without a pole at $\infty$]
\label{lem-Weilrational}
Let $Q=q^m$ with $q$ being a power of $2$. Let $f\in\FQ(X)$ have no pole at $\infty$, namely,
$$
f(X)=c_0+\frac{v(X)}{w(X)},\qquad c_0\in\FQ,\quad \deg v<\deg w,\quad \gcd(v,w)=1,
$$
and let $w(X)=\prod_{j=1}^{k}w_j(X)^{\,m_j}$ with $w_1(X),\dots,w_k(X)\in\FQ[X]$ being pairwise distinct, monic, and irreducible over $\FQ$. 
\begin{enumerate}[label=(\arabic*)]
\item Set
$$
L=\sum_{j=1}^{k}(m_j+1)\deg w_j
$$
and let $S$ be the set of poles of $f$ in $\FQ$. Assume that
\begin{enumerate}
\item[(1a)] $f$ is nonconstant
\item[(1b)] $m_j$ is odd for every $1\le j\le k$.
\end{enumerate}
Then
$$
\Big|\sum_{x\in\FQ\sm S}\varphi\big(f(x)\big)\Big|\ \le\ 1+(L-2)\sqrt Q .
$$
\item There exist $t \in \FQ[X]$ and integers $s_1,\dots,s_k \ge 0$ such that
$$
h=\frac{t(X)}{\prod_{j=1}^{k}w_j(X)^{s_j}}, \qquad \deg t<\sum_{j=1}^{k}s_j\deg w_j,
$$
where $h$ has no pole at $\infty$ and no pole at any monic irreducible other than $w_1,\dots,w_k$, such that
$$
\tilde f=f-(h^{2}+h)=\tilde c_0+\frac{\tilde v(X)}{\tilde w(X)}, \qquad \td{c_0} \in \FQ, \quad \deg \tilde v<\deg \tilde w, \quad \gcd(\tilde v,\tilde w)=1
$$
with $\tilde w(X)=\prod_{j=1}^{k}w_j(X)^{\tilde m_j}$ and with each $\tilde m_j \le m_j$ either $0$ or odd. Moreover, for every $x \in \FQ$ that is not a pole of $f$, we have
$$
\varphi\big(f(x)\big)=\varphi\big(\tilde f(x)\big).
$$
\end{enumerate}
\end{lemma}
\begin{proof}
Part (1) is a special case of \cite[Theorem~1.1]{CP06}. In the notation of \cite{CP06}:
\begin{enumerate}
\item[(a)] $M=0$ as $c_0 \in \FQ$. 
\item[(b)] The required conditions on $f$ follow from (1a) and (1b).  
\item[(c)] Take $g$ to be the constant polynomial $g(X) \equiv 1$, thus $U=0$. 
\item[(d)] Take the multiplicative character $\chi$ to be the principal character.
\end{enumerate}
Then by \cite[Theorem~1.1]{CP06}, we have
$$
\sum_{x\in\FQ\sm S}\varphi\big(f(x)\big)=\sum_{j=1}^{L-1} \omega_j^n,
$$
where among $1 \le j \le L-1$, there are exactly $L-2$ indices $j$ satisfying $|\omega_j|=q^{\frac{1}{2}}$ and one $j$ such that $|\omega_j|=1$. Consequently, the result follows. 

Part (2) is the characteristic $2$ case of the reduction in \cite[pp.~273--274]{CP06}.  We note that $f$ satisfies $v_{w_j}(f)=-m_j$ for $1 \le j \le k$, and $v_u(f) \ge 0$ for every other monic irreducible polynomial $u$. 

We use induction on $\deg w=\sum_j m_j \deg w_j$. If $\deg w=1$, then $w$ is an irreducible polynomial and it suffices to choose $h=0$. If every $m_j$ is odd, we can take $h=0$. Otherwise, fix some $1 \le j \le k$ with $m_j=2s$ even, and let $a \in \FQ[X]$, $\deg a<\deg w_j$, be the leading partial-fraction coefficient of $f$ at $w_j$. Therefore, $v_{w_j}\big(f-a/w_j^{m_j}\big)>-m_j$. Since $\FQ[X]/(w_j)$ is a finite field of characteristic $2$, there is $b \in \FQ[X]$ with $\deg b<\deg w_j$ and $b^{2} \equiv a \pmod{w_j}$. Put $h_j=b/w_j^{s}$. As $\deg b<\deg w_j \le s\deg w_j$, the function $h_j$ has no pole at $\infty$, $v_{w_j}(h_j)=-s>-m_j$, and $v_u(h_j) \ge 0$ for every monic irreducible $u \ne w_j$. Moreover
$w_j \mid a-b^{2}$, so $v_{w_j}\big(f-h_j^{2}\big)>-m_j$. Therefore $f-(h_j^{2}+h_j)$ has no pole at $\infty$, satisfies $v_{w_j}>-m_j$, and has the same valuation as $f$ at every other monic irreducible.

Therefore $f_1=f-(h_j^{2}+h_j)$ has no pole at $\infty$. Its poles lie among $w_1,\dots,w_k$ and $v_{w_i}(f_1)=v_{w_i}(f)=-m_i$ for $i \ne j$. Writing $m_i'=-\min\{v_{w_i}(f_1),0\}$, we thus have $m_i'=m_i$
for $i \ne j$ and $m_j'<m_j$, so $f_1$ satisfies the induction hypothesis with $\sum_i m_i'\deg w_i<\deg w$. Consequently, there exists $h' \in \FQ(x)$ such that 
$$
\tilde f=f_1-(h'^2+h')=\tilde c_0+\frac{\tilde v(X)}{\tilde w(X)}, \qquad \tilde w(X)=\prod_{j=1}^{k}w_j(X)^{\tilde m_j},
$$
satisfying the required conditions. Note that $\tilde f=f_1-(h'^2+h')=f-(h_j^{2}+h_j)-(h'^2+h')=f-((h_j+h')^2+(h_j+h'))$. Setting $h=h_j+h'$, the reduction from $f$ to $\tilde f$ is established. 

Finally, as every pole of $h$ is a pole of $f$, so $h(x)$ is defined at each $x \in \FQ$ that is not a pole of $f$. Since $\Tr(h(x)^{2}+h(x))=0$ for every $x \in \FQ$ that is not a pole of $f$, we have $\varphi\big(f(x)\big)=\varphi\big(\tilde f(x)\big)$.
\end{proof}

\begin{remark}
\label{rem-Weilrational}
Lemma \ref{lem-Weilrational}(2) indicates that condition (1b) does not constitute an essential obstruction. If $f$ fails condition (1b), then the reduction from $f$ to $\tilde f$ guarantees that $\tilde f$ satisfies condition (1b) and  $\varphi(\tilde f(x))$ agree with $\varphi(f(x))$ for every $x \in \FQ$ that is not a pole of $f$. It is noteworthy that the reduced function $\tilde f$ may be constant, which leads to the degenerate case $f=h^{2}+h+c_0$ with $c_0 \in \FQ$. This degenerate case requires separate treatment later.
\end{remark}

The following well-known lemma describes when a quadratic polynomial over $\Fq$ has $\Fq$-solutions. 

\begin{lemma}
\label{lem-quadeqn}
Consider the quadratic equation 
\begin{equation}
\label{eqn-quad}
x^2+ax+b=0
\end{equation}
with $a,b\in\Fq$.
\begin{enumerate}[label=(\arabic*)]
\item {\rm (\cite[Corollary 3.79]{LN97})} If $q$ is even and $a \ne 0$, then Equation \eqref{eqn-quad} has $\Fq$-solutions if and only if $\Trqt\left(\frac{b}{a^2}\right)=0$. 
\item {\rm (\cite[p. 130]{LN97})} If $q$ is odd, set $\Delta=a^2-4b$. Suppose Equation \eqref{eqn-quad} has $N$ distinct $\Fq$-solutions. Then
$$
N=\begin{cases}
     0 & \mbox{if $\eta(\Delta)=-1$,} \\
     1 & \mbox{if $\eta(\Delta)=0$,} \\
     2 & \mbox{if $\eta(\Delta)=1$.}
  \end{cases} 
$$
\end{enumerate}
\end{lemma}

We will also need the following lemma describing the number of rational points on an elliptic curve.

\begin{lemma}[{\cite[Theorem 4.12]{W08}}]
\label{lem-elliptic}
Let $E$ be an elliptic curve over $\Fq$ and $\#E(\Fq)$ be the number of $\Fq$-rational points on $E$. Let $\#E(\Fq)=q+1-a$. Write $X^2-aX+q=(X-\al)(X-\be)$. Then $\al+\be=a$, $\al\be=q$, and 
$$
\#E(\F_{q^m})=q^m+1-(\al^m+\be^m)
$$
for each $m \ge 1$.
\end{lemma}

\subsection{Generalized Hamming weights and generalized Supercode Lemma}

For $\bc=(c_1,c_2,\ldots,c_n) \in \Fq^n$, we define the \emph{support} of $\bc$ to be
$$
\supp(\bc)=\left\{ i \in [n] \mid c_i \ne 0 \right\}.
$$
For a subset $D \subset \Fq^n$, define the support of $D$ as
$$
\supp(D)=\bigcup_{\bc \in D} \supp(\bc).
$$
Below, we describe the notion of generalized Hamming weights of a linear code, which was introduced in \cite{Wei91} as an extension of the minimum distance.

\begin{definition}[generalized Hamming weights]
\label{def-ghw}
Let $\cC$ be an $[n,k]_q$ linear code. For any positive integer $1 \le r \le k$, the $r$-th generalized Hamming weight of $\cC$, denoted by $d_r(\cC)$, is defined as the minimum support size of an $r$-dimensional subcode of $\cC$:
$$
d_r(\cC)=\min \left\{ \left|\supp(D)\right| \mid D \le \cC, \dim_{\Fq}(D)=r \right\}.
$$
\end{definition}

\begin{remark}\label{rem-ghw}
Let $\cC$ be an $[n,k,d]_q$ linear code.
\begin{enumerate}[label=(\arabic*)]
\item For $1 \le r \le k$, the $r$-th generalized Hamming weight $d_r(\cC)$ is the smallest support size among all $r$-dimensional subspaces of $\cC$ over $\Fq$. In particular, the first generalized Hamming weight $d_1(\cC)$ is exactly the minimum distance $d$ of $\cC$.
\item Similar to the generalized covering radius, the following monotonicity holds:
$$
1 \le d = d_1(\cC) < d_2 (\cC) < \cdots < d_k(\cC) \le n.
$$
\end{enumerate}
\end{remark}

We will need the following extension of the Supercode Lemma \cite[Proposition 1]{CKMS85}, in which generalized Hamming weights provide lower bounds on the generalized covering radii. 

\begin{proposition}[Generalized Supercode Lemma {\cite[Lemma III.1]{XY26}}] \label{prop-supercode}
Let $q$ be a prime power. Let $\cC$ and $\cC'$ be linear codes over $\Fq$ such that $\cC \subsetneq \cC' \subseteq \mathbb{F}_q^n$. For any $r \ge 1$, if $\dim(\cC') - \dim(\cC) \ge r$, then
$$
\rho_r(\cC) \ge d_r(\cC,\cC') \ge d_r(\cC').
$$
Here $d_r(\cC')$ is the $r$-th generalized Hamming weight of $\cC'$, $\rho_r(\cC)$ is the $r$-th generalized covering radius of $\cC$, and $d_r(\cC,\cC')$ is defined as
$$
d_r(\cC,\cC') = \max_{(\mathbf{c}'_1, \dots, \mathbf{c}'_r) \in \mathcal{T}_r} \min_{(\mathbf{c}_1, \dots, \mathbf{c}_r) \in C^r} \left| \bigcup_{i=1}^r \supp(\mathbf{c}'_i - \mathbf{c}_i) \right|,
$$
where $\mathcal{T}_r$ is the set of all $r$-tuples $(\mathbf{c}'_1, \dots, \mathbf{c}'_r) \in (\cC')^r$ that are linearly independent modulo $\cC$.
\end{proposition}

In this next subsection, we will provide detailed descriptions of Melas codes and their supercodes that lead to lower bounds on the generalized covering radii of Melas codes in view of the Generalized Supercode Lemma.

\subsection{Irreducible primitive cyclic codes and Melas codes}

Now we fix some notation that will be used throughout the rest of the paper. Let $q$ be a prime power and $m \ge 1$. Set $Q=q^m$. Let $\al$ be a primitive element of $\FQ$. For $0 \le i \le Q-2$, denote the minimal polynomial of $\al^i$ over $\Fq$ by $m_{\al^i}(x)$.  An \emph{irreducible primitive cyclic code} over $\Fq$ with one single zero $\al$, denoted by $\cC_{1}^{m,q}$, has generator polynomial $m_{\al}(x)$. Regard both $\Fqm$ and $\Fq^m$ as $m$-dimensional vector space over $\Fq$. Let 
$$
\phi: \Fqm \rightarrow \Fq^m
$$ 
be an $\Fq$-linear isomorphism. Let
$$
H_1^{m,q}=
\begin{bmatrix}
1 & \al & \al^{2} & \cdots & \al^{Q-2}
\end{bmatrix}.
$$
Then the parity-check matrix of $\cC_1^{m,q}$ is
$$
\phi(H_1^{m,q})=
\begin{bmatrix}
\phi(1) & \phi(\al) & \phi(\al^{2}) & \cdots & \phi(\al^{Q-2})
\end{bmatrix}.
$$
Moreover, $\cC_1^{m,q}$ is a $[q^m-1,q^m-1-m]_q$ code. When $q=2$, $\cC_1^{m,2}$ is the binary Hamming code \cite[Section 1.8]{HP03}. For odd prime power $q$, set 
$$
\wt{H_1^{m,q}}=
\begin{bmatrix}
1 & \al & \al^{2} & \cdots & \al^{\frac{Q-3}{2}}
\end{bmatrix}.
$$
Let $\wt{\cC_1^{m,q}}$ denote the $[\frac{q^m-1}{2},\frac{q^m-1}{2}-m]_q$ code with parity-check matrix
$$
\phi(\wt{H_1^{m,q}})=
\begin{bmatrix}
\phi(1) & \phi(\al) & \phi(\al^{2}) & \cdots & \phi(\al^{\frac{Q-3}{2}})
\end{bmatrix}.
$$
When $q=3$, $\wt{\cC_1^{m,3}}$ is the ternary Hamming code \cite[Section 1.8]{HP03}.

In order to apply the Generalized Supercode Lemma, we need the following result which describes the generalized Hamming weight of irreducible primitive cyclic codes and related codes.

\begin{proposition}
\label{prop-GHW}
\begin{enumerate}[label=(\arabic*)]
\item Let $m \ge 2$. Then
$$
\left\{d_r(\cC_1^{m,q}) \mid 1 \le r \le q^m-1-m\right\}=\left\{ 1 \le i \le q^m-1 \mid i \notin \{ 1,q,q^2,\ldots,q^{m-1}\} \right\}.
$$
Therefore, for $2 \le s \le q^m-1$, if $s$ is not a power of $q$, then
$$
d_r(\cC_1^{m,q})=s,
$$ 
where $r=s-\lc \log_q s \rc$. Consequently, for $1 \le r \le q^m-1-m$, we have
$$
r+\lc \log_q (r+1) \rc \le d_r(\cC_1^{m,q}) \le r+\lc \log_q (r+m) \rc.
$$
\item Let $m \ge 2$. Then 
$$
\{d_r(\wt{\cC_1^{m,3}}) \mid 1 \le r \le \frac{3^m-1}{2}-m\}=\{ 1 \le i \le \frac{3^m-1}{2} \mid i \notin \{ \frac{3^j+1}{2} \mid 0 \le j \le m-1 \}\}.
$$
Therefore, for $3 \le s \le \frac{3^m-1}{2}$, if $2s-1$ is not a power of $3$, then
$$
d_r(\wt{\cC_1^{m,3}})=s,
$$ 
where $r=s-\lc \log_3 (2s-1) \rc$. Consequently, for $1 \le r \le \frac{3^m-1}{2}-m$, we have
$$
r+\lc \log_3 (2r+3) \rc \le d_r(\wt{\cC_1^{m,3}}) \le r+\lc \log_3 (2r+2m-1) \rc.
$$
\end{enumerate}
\end{proposition}
\begin{proof}
(1) According to \cite[Corollary 7]{YLFL15}, the generalized Hamming weights of the dual code of $\cC_1^{m,q}$ is $d_r((\cC_1^{m,q})^\perp)=q^m-q^{m-r}$, where $0 \le r \le m$. By \cite[Theorem 3]{Wei91}, the generalized Hamming weights of $\cC_1^{m,q}$ follows. Specifically,  
$$
\left\{d_r(\cC_1^{m,q}) \mid 1 \le r \le q^m-1-m\right\}=\left\{ 1 \le i \le q^m-1 \mid i \notin \{ 1,q,q^2,\ldots,q^{m-1}\} \right\}.
$$
Since $d_r(\cC_1^{m,q})$ precisely skips the powers of $q$ within $[q^m-1]$, we have $d_r(\cC_1^{m,q})=s$ with $r=s-\lc \log_q s \rc$. Since $s=r+\lc \log_q s \rc=r+\lc \log_q (r+\lc \log_q s \rc) \rc$ and $2 \le s \le q^m-1$, therefore,
$$
r+\lc \log_q (r+1) \rc \le d_r(\cC_{1}^{m,q})=s \le r+\lc \log_q (r+m) \rc.
$$

(2) Note that $\wt{\cC_1^{m,3}}$ is the ternary Hamming code, which is the dual code of the ternary projective first order Reed-Muller code. According to \cite[Proposition 7.1 b)]{TV95}, the generalized Hamming weight $d_r((\wt{\cC_{1}^{m,3}})^\perp)=\frac{3^m-3^{m-r}}{2}$ with $1 \le r \le m$.  By \cite[Theorem 3]{Wei91}, the generalized Hamming weights of $\wt{\cC_{1}^{m,3}}$ follows. Specifically, 
$$
\{d_r(\wt{\cC_1^{m,3}}) \mid 1 \le r \le \frac{3^m-1}{2}-m\}=\{ 1 \le i \le \frac{3^m-1}{2} \mid i \notin \{ \frac{3^j+1}{2} \mid 0 \le j \le m-1 \}\}.
$$
The lower and upper bounds on $d_r(\wt{\cC_{1}^{m,3}})$ follow from a straightforward computation.  
\end{proof}

The \emph{Melas code}, denoted by $M(m,q)$, is a cyclic code over $\Fq$ with generator polynomial 
$$
g_{m,q}(x)=\lcm \left(m_{\al}(x),m_{\al^{-1}}(x)\right).
$$ 
Set 
$$
H_{1,-1}^{m,q}=
\begin{bmatrix}
1 & \al & \al^2 & \cdots & \al^{Q-2} \\
1 & \al^{-1} & \al^{-2} & \cdots & \al^{-(Q-2)}
\end{bmatrix}.
$$
Then $M(m,q)$ has parity-check matrix
$$
\phi(H_{1,-1}^{m,q})=
\begin{bmatrix}
\phi(1) & \phi(\al) & \phi(\al^2) & \cdots & \phi(\al^{Q-2}) \\
\phi(1) & \phi(\al^{-1}) & \phi(\al^{-2}) & \cdots & \phi(\al^{-(Q-2)})
\end{bmatrix}.
$$
Clearly, $\cC_1^{m,q}$ is a supercode of $M(m,q)$. Moreover, define 
$$
\wt{H_{1,-1}^{m,3}}=
\begin{bmatrix}
1 & \al & \al^2 & \cdots & \al^{\frac{Q-3}{2}} \\
1 & \al^{-1} & \al^{-2} & \cdots & \al^{-\frac{Q-3}{2}}
\end{bmatrix}.
$$
and let $\wt{M(m,3)}$ be the linear code with parity-check matrix 
$$
\phi(\wt{H_{1,-1}^{m,3}})=
\begin{bmatrix}
\phi(1) & \phi(\al) & \phi(\al^2) & \cdots & \phi(\al^{\frac{Q-3}{2}}) \\
\phi(1) & \phi(\al^{-1}) & \phi(\al^{-2}) & \cdots & \phi(\al^{-\frac{Q-3}{2}})
\end{bmatrix}.
$$
Therefore, $\wt{M(m,3)}$ is a $[\frac{3^m-1}{2},\frac{3^m-1}{2}-2m]_3$ code and $\wt{\cC_1^{m,3}}$ is a supercode of $\wt{M(m,3)}$. To compute the generalized cover radius, it is more convenient to study the $\Fq$-span of the columns of $H_{1,-1}^{m,q}$ and $\wt{H_{1,-1}^{m,3}}$, rather than working directly with the parity-check matrices $\phi(H_{1,-1}^{m,q})$ and $\phi(\wt{H_{1,-1}^{m,3}})$.

The covering radii of Melas codes have been thoroughly studied in a series of literature, for which we refer to \cite{SHOS22} for a comprehensive account and the references therein. The covering radii of Melas code $M(m,q)$ have been completely determined from all $(m,q)$ pairs. 

\begin{result}
\label{res-cr}
Let $m \ge 1$ and $q$ be a prime power. Then the following holds
$$
\rho_1(M(m,q))=
\begin{cases}
1 & \mbox{if $q=2$, $m \in \{1,2\}$} \\
3 & \mbox{if $q=2$, $m \ge 3$} \\
1 & \mbox{if $q=3$, $m = 1$} \\
4 & \mbox{if $q=3$, $m = 2$} \\
3 & \mbox{if $q=3$, $m \ge 3$} \\
2 & \mbox{if $q \ge 4$, $m \ge 1$} 
\end{cases}
$$
\end{result}

\begin{remark}\label{rem-Melas}
\begin{enumerate}[label=(\arabic*)]
\item Unless $(m,q) \in \{(2,2),(1,2),(1,3)\}$, $m_{\al}(x) \ne m_{\al^{-1}}(x)$ and therefore, the Melas code $M(m,q)$ has the generator polynomial $g_{m,q}(x)=m_{\al}(x)m_{\al^{-1}}(x)$. Thus, the Melas code $M(m,q)$ always has parameters 
$$
\begin{cases}
[q^m-1,q^m-1-m]_q & \mbox{if $(m,q) \in \{ (1,2), (1,3), (2,2) \}$} \\
[q^m-1,q^m-1-2m]_q & \mbox{otherwise}
\end{cases}
$$
Specifically, $M(1,2)$ is a degenerate $[1,0]_2$ code with parity-check matrix 
$$
\phi(H_{1,-1}^{1,2})=
\begin{bmatrix}
1 \\
\end{bmatrix}
$$
and consisting of the zero vector as the only codeword. Thus, $\rho_1(M(1,2))=1$. Moreover, in view of Remark \ref{rem-gcr}(2), $\rho_2(M(1,2))=1$.

$M(1,3)$ is a $[2,1,2]_3$ code with parity-check matrix 
$$
\phi(H_{1,-1}^{1,3})=
\begin{bmatrix}
2 & 1 \\
\end{bmatrix}.
$$
Thus, $\rho_1(M(1,3))=1$. Moreover, in view of Remark \ref{rem-gcr}(2), $\rho_2(M(1,3))=1$.

$M(2,2)$ is a  $[3,1,3]_2$ code with parity-check matrix 
$$
\phi(H_{1,-1}^{2,2})=
\begin{bmatrix}
0 & 1 & 1 \\
1 & 0 & 1 \\
\end{bmatrix}.
$$
Thus, $\rho_1(M(2,2))=1$ and $\rho_2(M(2,2))=2$.

We have so far obtained $\rho_2(M(m,q))$ for $(m,q) \in \{ (1,2), (1,3), (2,2) \}$ in Equation \eqref{eqn-seccr}. Below, when considering the second generalized covering radii, we will exclude the known cases of $(m,q) \in \{ (1,2), (1,3), (2,2) \}$ for the sake of simplicity.
\item In the literature, for instance \cite[p. 4354]{SHOS22}, it is mentioned that $\rho_1(M(2,2))=3$. On the other hand, in view of the parity-check matrix $\phi(H_{1,-1}^{2,2})$, its three columns consists of all nonzero vectors in $\Ft^2$, thus the covering radius $\rho_1(M(2,2))$ should be equal to $1$.
\end{enumerate}
\end{remark}

\section{Lower and upper bounds on the generalized covering radii of Melas codes} \label{sec-bound}
\subsection{Lower bounds}

In this subsection, we derive several lower bounds on $\rho_t(M(m,q))$. 

Employing the Generalized Supercode Lemma in Proposition \ref{prop-supercode} and the generalized Hamming weight of irreducible primitive cyclic codes and related codes in Proposition \ref{prop-GHW}, we obtain the following lower bound on $\rho_t(M(m,q))$.

\begin{theorem}
\label{thm-GHWlb}
\begin{enumerate}[label=(\arabic*)]
\item Let $m \ge 2$ with $(m,q) \ne (2,2)$, and let 
$$1 \le t \le \min\{m,q^m-1-m\}.$$ 
Then 
$$\rho_t(M(m,q)) \ge d_t(\cC_1^{m,q}) \ge t+\lc \log_q (t+1) \rc.$$
\item Let $m \ge 2$ and $1 \le t \le m$. Then 
$$\rho_t(M(m,3)) \ge d_t(\wt{\cC_1^{m,3}}) \ge t+\lc \log_3 (2t+3) \rc.$$
\end{enumerate}
\end{theorem}
\begin{proof}
(1) As $m \ge 2$ with $(m,q) \ne (2,2)$, by Remark \ref{rem-Melas}(1), we have $\dim_{\Fq}(\cC_1^{m,q})-\dim_{\Fq}(M(m,q))=m$. Then the result follows directly from Proposition \ref{prop-GHW}(1).

\noindent (2) First recall that $\wt{\cC_1^{m,3}}$ is a supercode of $\wt{M(m,3)}$. As $\dim_{\F_3}(\wt{\cC_1^{m,3}})-\dim_{\F_3}(\wt{M(m,3)})=m$, applying Proposition \ref{prop-supercode} and Proposition \ref{prop-GHW}(2), we have
$\rho_t(\wt{M(m,3)}) \ge d_t(\wt{\cC_1^{m,3}}) \ge t+\lc \log_3 (2t+3) \rc$. Moreover, note that for $0 \le j \le \frac{Q-3}{2}$,
we have
$$
\begin{bmatrix}
\al^{j+\frac{Q-1}{2}} \\
\al^{-j-\frac{Q-1}{2}}
\end{bmatrix}
=
-\begin{bmatrix}
\al^{j} \\
\al^{-j}
\end{bmatrix}
$$
Therefore, the $j$-th column and the $(j+\frac{Q-1}{2})$-th column of $M(m,3)$ generate the same one-dimensional subspace over $\F_3$. Thus, either column may be used in forming an $\F_3$-span without changing the resulting subspace. It follows from Definition~\ref{def-gcr1} that, for every $t$, the $t$-th generalized covering radius of $M(m,3)$ is equal to that of $\wt{M(m,3)}$, namely, $\rho_t(M(m,3))=\rho_t(\wt{M(m,3)})$. This completes the proof. 
\end{proof}

Below, we observe that the combinatorial argument in \cite[Theorem IV.1]{XY26} can be applied to give the following stronger lower bound on $\rho_t(M(m,q))$ when $m$ is sufficiently large compared with $t$, which justifies Theorem \ref{thm-asymlb}.

\begin{theorem}
Let $q$ be a prime power and $m \ge 1$. Suppose $(m,q) \notin \{(2,2), (1,2), (1,3)\}$. For $1 \le t \le m$, if 
$$	
q^m \ge \frac{q^{t(2t-1)}}{(2t-1)!},
$$
then $\rho_t(M(m,q)) \geq 2t$. In particular, if $m \ge t(2t-1)$ and $(m,q) \notin \{(2,2), (1,2), (1,3)\}$, then $\rho_t(M(m,q)) \geq 2t$.
\end{theorem}
\begin{proof}
For each $x \in \F_{q^m}^*$, define a vector $\bv(x)=\qbinom{x}{x^{-1}} \in \Fqm^2$. Suppose $\rho_t(M(m,q)) \le 2t-1$, then for any $\left(\ba_1,\ba_2,\cdots,\ba_t\right) \in \left(\Fqm^2\right)^t$, there exist $x_1,x_2,\cdots,x_{2t-1} \in \F_{q^m}^*$ such that 
$$
(\ba_1,\ba_2,\cdots,\ba_t) \in \left(\spa_{\Fq}\{ \bv(x_1), \bv(x_2), \ldots, \bv(x_{2t-1})\}\right)^t.
$$	
Consequently,
$$
\left(\Fqm^2\right)^t \subseteq \bigcup_{\{x_1,x_2,\cdots,x_{2t-1}\} \subseteq \Fqm^*} \left(\spa_{\Fq} \{\bv(x_1), \bv(x_2),\cdots, \bv(x_{2t-1})\}\right)^t.
$$
By counting the size of the above two subsets, we have
$$
q^{2tm} \le \binom{q^m-1}{2t-1} \left(q^{2t-1}\right)^t < \frac{(q^m)^{2t-1}}{(2t-1)!}q^{t(2t-1)},
$$
which implies that 
$$
q^m < \frac{q^{t(2t-1)}}{(2t-1)!}.
$$
Consequently, if
$$
q^m \ge \frac{q^{t(2t-1)}}{(2t-1)!},
$$
then
$$
\rho_t(M(m,q)) \ge 2t.
$$
\end{proof}

For the second generalized covering radius, we have the following lower bound that holds for all Melas codes $M(m,q)$, except for $(m,q) \in \{ (1,2), (1,3), (2,2) \}$, which matches the lower bound in Theorem \ref{thm-asymlb} without the condition on $m$.  

\begin{theorem}\label{thm-gcv2lb}
Let $q$ be a prime power and $m \ge 2$. Then $\rho_2(M(m,q)) \ge 4$, except for $(m,q) =(2,2)$.
\end{theorem}
\begin{proof}
Let $\be_1$ and $\be_2$ be distinct elements of $\Fqm^*$ such that $\frac{\be_2}{\be_1} \notin \Fq^*$. Assume otherwise that $\rho_2(M(m,q)) \le 3$. Then there exist distinct $\bx_1, \bx_2 \in \Fq^{Q-1} \sm \{\bz\}$, regarded as column vectors, such that for $i \in \{1,2\}$,
\begin{equation}
\label{eqn-pc}
H_{1,-1}^{m,q} \bx_i=\begin{bmatrix}
                         0 \\
                         \be_i
                     \end{bmatrix}
\end{equation}
and $|\supp(\bx_1) \cup \supp(\bx_2)| \le 3$. In view of Equation \eqref{eqn-pc}, for $i \in \{1,2\}$, we can regard $\supp(\bx_i)$ as a subset of $\{1,\al,\al^2,\cdots,\al^{Q-2}\}$. Since $H_1^{m,q}\bx_i=0$ for $i \in \{1,2\}$, we have $\bx_1, \bx_2 \in \cC_1^{m,q}$. Note that 
$$
d(\cC_1^{m,q})=\begin{cases}
                 3 & \mbox{if $q=2$,} \\
                 2 & \mbox{if $q>2$.}
               \end{cases}
$$
Thus, for $i \in \{1,2\}$,
$$
\wth(\bx_i) \ge \begin{cases}
                 3 & \mbox{if $q=2$,} \\
                 2 & \mbox{if $q>2$.}
               \end{cases}
$$
Since $|\supp(\bx_1) \cup \supp(\bx_2)| \le 3$, then $\wth(\bx_i) \le 3$ for $i \in \{1,2\}$. 

If $q=2$, then $\wth(\bx_i)=3$ for $i \in \{1,2\}$ and $\supp(\bx_1)=\supp(\bx_2)$. Thus, $\bx_1=\bx_2$, which is impossible in view of Equation \eqref{eqn-pc} as $\be_1 \ne \be_2$. 

If $q>2$, then $2 \le \wth(\bx_i) \le 3$ for $i \in \{1,2\}$. Since $|\supp(\bx_1) \cup \supp(\bx_2)| \le 3$, then $|\supp(\bx_1) \cup \supp(\bx_2)| \in \{ 2,3 \}$. If $|\supp(\bx_1) \cup \supp(\bx_2)|=2$, then $\wth(\bx_1)=\wth(\bx_2)=2$ and $\supp(\bx_1)=\supp(\bx_2)$. Assume that $\supp(\bx_1)=\supp(\bx_2)=\{ \al^j, \al^\ell \}$. Then there exist $a_{1j}, a_{1\ell}, a_{2j}, a_{2\ell} \in \Fq^*$ such that for $i \in \{1,2\}$,
$$
\begin{cases}
a_{ij} \al^j+a_{i\ell}\al^{\ell}=0  \\
a_{ij} \al^{-j}+a_{i\ell}\al^{-\ell}=\be_i 
\end{cases}
$$
Consequently, $\al^{j-\ell}=-\frac{a_{i\ell}}{a_{ij}} \in \Fq^*$ and 
$$
\be_i=\al^{-j}\left(a_{ij}+a_{i\ell}\al^{j-\ell}\right)
     =\al^{-j}\left(a_{ij}-\frac{a_{i\ell}^2}{a_{ij}}\right).
$$
Note that $a_{ij}-\frac{a_{i\ell}^2}{a_{ij}}=\frac{(a_{ij}-a_{i\ell})(a_{ij}+a_{i\ell})}{a_{ij}} \in \Fq$. We claim that $a_{ij}-\frac{a_{i\ell}^2}{a_{ij}}$ is nonzero. If $a_{ij}=-a_{i\ell}$, then $\alpha^{j-\ell}=-a_{i\ell}/a_{ij}=1$, so $j=\ell$, contradiction. If $a_{ij}=a_{i\ell}$, then $\alpha^{j-\ell}=-1$ and $\beta_i = a_{ij}\alpha^{-j}+a_{i\ell}\alpha^{-\ell}= a_{ij}\alpha^{-j}\bigl(1+\alpha^{\,j-\ell}\bigr)=0$,
contradicting $\beta_i\ne0$. Hence, $\be_i \al^{-j}(a_{ij}-\frac{a_{i\ell}^2}{a_{ij}}) \in \al^{-j}\Fq^*$ for $i \in \{1,2\}$. This forces $\frac{\be_2}{\be_1} \in \Fq^*$, which is impossible. If $|\supp(\bx_1) \cup \supp(\bx_2)|=3$, then $2 \le \wth(\bx_1), \wth(\bx_2) \le 3$ and $\supp(\bx_1),\supp(\bx_2) \subseteq \{ \al^j, \al^{\ell}, \al^s \}$. Since $\supp(\bx_1) \cap \supp(\bx_2) \ne \es$, without loss of generality, we can assume that $\al^s \in \supp(\bx_1) \cap \supp(\bx_2)$, $\al^j \in \supp(\bx_1)$, and $\al^\ell \in \supp(\bx_2)$. Then there exist $a_{1j}, a_{1s}, a_{2\ell}, a_{2s} \in \Fq^*$ and $a_{1\ell}, a_{2j} \in \Fq$ such that for $i \in \{1,2\}$,
$$
\begin{cases}
a_{ij} \al^j+a_{i\ell}\al^{\ell}+a_{is}\al^s=0  \\
a_{ij} \al^{-j}+a_{i\ell}\al^{-\ell}+a_{is}\al^{-s}=\be_i 
\end{cases}
$$
Note that there exists $\la \in \Fq^*$ such that $a_{2s}=\la a_{1s}$ and
$$
(a_{2j}-\la a_{1j})\al^j+(a_{2\ell}-\la a_{1\ell})\al^\ell=0.
$$
If $a_{2\ell}=\la a_{1\ell}$, then $a_{2j}=\la a_{1j}$. This implies $\be_2=\la \be_1$, which is impossible. If $a_{2\ell} \ne \la a_{1\ell}$, then $a_{2j} \ne \la a_{1j}$. We have
\begin{align*}
\al^{j-\ell}&=-\frac{a_{2\ell}-\la a_{1\ell}}{a_{2j}-\la a_{1j}} \in \Fq^* \\
\al^{s-\ell}&=a_{1s}^{-1}(a_{1j}\frac{a_{2\ell}-\la a_{1\ell}}{a_{2j}-\la a_{1j}}-a_{1\ell}) \in \Fq^*
\end{align*}
Therefore,
$$
\be_i=a_{ij} \al^{-j}+a_{i\ell}\al^{-\ell}+a_{is}\al^{-s}=\al^{-\ell}(a_{ij}\al^{\ell-j}+a_{i\ell}+a_{is}\al^{\ell-s}) \in \al^{-\ell}\Fq^*.
$$
This forces $\frac{\be_2}{\be_1} \in \Fq^*$, which is impossible. In summary, the assumption that $\rho_2(M(m,q)) \le 3$ always leads to contradiction and we must have $\rho_2(M(m,q)) \ge 4$.
\end{proof}

\subsection{Upper bounds}

In this subsection, we will prove Theorem \ref{thm-asymub}. We first describe the following auxiliary lemma that narrows the range of vectors we need to consider to establish the upper bound $\rho_t(M(m,q)) \le 2t+1$.  

\begin{lemma}
\label{lem-normalization}
Let $q$ be a prime power, $1 \le t \le 2m$, and $Q=q^m$. Suppose that for every $\bs_1,\dots,\bs_{t}\in\FQ^{2}$ with $\bs_i=\qbinom{a_i}{b_i}$ satisfying
\begin{enumerate}[label=\rm(\arabic*)]
\item $\bs_1,\dots,\bs_{t}$ are linearly independent over $\Fq$,
\item $b_i\ne0$ for every $1\le i\le t$,
\end{enumerate}
there exists $X\subseteq\FQ^{*}$ with $|X|\le2t+1$ and $\{\bs_1,\dots,\bs_{t} \} \subseteq \spa_{\Fq}\{\bv(x) \mid x\in X\}$, where $\bv(x)=\qbinom{x}{x^{-1}}$. Then
$$
\rho_t\big(M(m,q)\big) \le 2t+1 .
$$
\end{lemma}
\begin{proof}
Write $U=\left\{ \qbinom{a}{0} \;\middle|\; a \in \FQ \right\}$, which is an $\Fq$-subspace of $\FQ^2$ with $\dim_{\Fq}U=m$. Let $\tau:\FQ^2 \to \FQ^2$ be the coordinate swap $\tau\big(\qbinom{a}{b}\big)=\qbinom{b}{a}$, an $\Fq$-linear involution.

Let $\bw_1, \bw_2, \ldots,\bw_t$ be arbitrary $t$ vectors in $\FQ^2$. Set $W=\spa_{\Fq}\{\bw_1, \bw_2, \ldots, \bw_t\}$ and $\ell=\dim_{\Fq}W \le t$. Since the span of a set of columns is an $\Fq$-subspace, it suffices to generate $X \subseteq \FQ^*$ satisfying $|X| \le 2t+1$ and $W \subseteq \spa_{\Fq}\{\bv(x):x \in X\}$. If $\ell=0$, then $W=\{ \bf{0} \}$ and the conclusion is trivially true. Thus, we can assume $\ell>0$ and $W \ne \{ \bf{0} \}$. 

We may assume $W \not\subseteq U$. Suppose $W \subseteq U$. Then $\tau(W) \subseteq \tau(U)=\left\{ \qbinom{0}{b} \mid b \in \FQ \right\}$. Thus $\tau(W) \cap U=\{\bz\}$ and in particular $\tau(W) \not\subseteq U$. Note that if $\tau(W) \subseteq \spa_{\Fq}\{\bv(x) : x \in X\}$, then $W \subseteq \spa_{\Fq}\{\bv(x) : x \in \{ y^{-1} \mid y \in X \} \}$. Therefore we can replace $W$ by $\tau(W)$ if needed and therefore may assume that $W \not\subseteq U$.

We claim that $W$ has an $\Fq$-basis $\{\bw_1,\dots,\bw_\ell\}$ with $\bw_k \notin U$ for every $1 \le k \le \ell$. As $W \not\subseteq U$, without loss of generality, we can assume $\bw_1 \notin U$. For $1 \le i \le \ell$, define $\bw_1'=\bw_1$ and
$$
\bw_i'=\begin{cases}
              \bw_1+\bw_i & \mbox{if $\bw_i \in U$} \\
              \bw_i             & \mbox{if $\bw_i \notin U$}
           \end{cases}
$$
Therefore, $\{ \bw_i' \mid 1 \le i \le \ell \}$ is an $\Fq$-basis of $W$ such that $\bw_i' \notin U$ for each $1 \le i \le t$. We can extend the basis $\{\bw_1',\dots,\bw_\ell'\}$ to a linearly independent set of size $t$ written as $\{\bs_1, \bs_2, \ldots,\bs_t\}$. Following the same idea, we can assume $\bs_i \notin U$ for every $1 \le i \le t$. Thus, by the condition of the lemma, there exists $X \subseteq \FQ^*$ with $|X| \le 2t+1$ such that  $\{\bs_1, \bs_2, \ldots, \bs_t\} \subseteq \spa_{\Fq}\{\bv(x):x \in X\}$. Consequently, $W \subseteq \spa_{\Fq}\{\bs_1, \bs_2, \ldots, \bs_t\} \subseteq \spa_{\Fq}\{\bv(x):x \in X\}$. Therefore, $\rho_t\big(M(m,q)\big) \le 2t+1$.
\end{proof}

In order to show $\rho_t\big(M(m,q)\big) \le 2t+1$, by Lemma \ref{lem-normalization}, it suffices to show that for an arbitrary independent set of vectors $\{ \bs_1, \bs_2, \ldots, \bs_t \}$ with $\bs_i=\qbinom{a_i}{b_i} \in \FQ^2$ and $b_i \ne 0$, $1\le i\le t$, there exists a set $X \subset \FQ^*$, such that $|X| \le 2t+1$ and $\{ \bs_1, \bs_2, \ldots, \bs_t \} \subseteq \spa_{\Fq}\{ \bv(u) \mid u \in X\}$. Moreover, considering the subset $X=\{x\} \cup \{ y_i \mid 1 \le i \le t\} \cup \{ z_i \mid 1 \le i \le t\} \subset \FQ^*$, it suffices to show that for each $1 \le i \le t$, the following system holds
\begin{equation}
\label{eqn-sys}
\begin{cases} 
a_i=x+y_i+z_i \\
b_i=x^{-1}+y_i^{-1}+z_i^{-1}
\end{cases}
\end{equation}
We choose $x \in \FQ^*$ such that $x \ne a_i$ and $x \ne b_i^{-1}$ for each $1 \le i \le t$. Set
\begin{equation}
\label{eqn-ABE}
A=\{a_i \mid 1\le i\le t\},\qquad B=\{b_i^{-1} \mid 1\le i\le t\},\qquad E=\{0\}\cup A\cup B
\end{equation}
Then $x \notin E$ and $|E| \le 2t+1$. Depending on the parity of $q$, we handle the $q$ even and $q$ odd cases separately below. The following proposition establishes Theorem \ref{thm-asymub} for $q$ even.

\begin{proposition}
\label{prop-even}
Let $q$ be an even prime power, $1 \le t \le m$, and $Q=q^m$. Suppose $\sqrt Q \ge t2^{t+1}$. Then $\rho_t(M(m,q)) \le 2t+1$. 
\end {proposition}
\begin{proof}
We use the notation in the paragraph containing System \eqref{eqn-sys} and Equation \eqref{eqn-ABE}. Set $\Tr=\TrQt$. For $x\in\FQ\sm E$, System \eqref{eqn-sys} is solvable with $y_i,z_i\in\FQ^{*}$ if and only if
\begin{equation}
\label{eqn-quad-even}
y_i^{2}+(x+a_i)\,y_i+\frac{x(x+a_i)}{b_ix+1}=0
\end{equation}
has an $\FQ$-solution $y_i$. Indeed, as $x \notin \{0,a_i,b_i^{-1}\}$, the term $\frac{x(x+a_i)}{b_ix+1}$ is nonzero. Thus, $y_i \notin \{ 0,x+a_i\}$ and $z_i=x+y_i+a_i \ne 0$.  

By Lemma \ref{lem-quadeqn}(1), Equation \eqref{eqn-quad-even} has an $\FQ$-solution if and only if
$$
\Tr\Big(\frac{x}{(b_ix+1)(x+a_i)}\Big)=0.
$$
Define 
$$
N=|\{x \in\FQ\sm E \mid \Tr\Big(\frac{x}{(b_ix+1)(x+a_i)}\Big)=0 \; \mbox{for each $1\le i \le t$} \}|.
$$
By Proposition~\ref{prop-evenproof} in Appendix A, we have $N>0$ and therefore, the proof is complete.
\end{proof}

The following proposition establishes Theorem \ref{thm-asymub} for $q$ odd and $t \ge 2$. 

\begin{proposition}
\label{prop-odd}
Let $q$ be an odd prime power, $2 \le t \le m$, and $Q=q^m$. Suppose $\sqrt Q \ge t2^{t+1}$. Then $\rho_t(M(m,q)) \le 2t+1$. 
\end {proposition}
\begin{proof}
We use the notation in the paragraph containing System \eqref{eqn-sys} and Equation \eqref{eqn-ABE}. Let $\eta$ denote the quadratic character of $\FQ$, extended to all of $\FQ$ by $\eta(0)=0$. For $x\in\FQ\sm E$, System \eqref{eqn-sys} is solvable with $y_i,z_i\in\FQ^{*}$ if and only if 
\begin{equation}
\label{eqn-quad-odd}
y_i^{2}+(x-a_i)\,y_i+\frac{x(a_i-x)}{b_ix-1}=0
\end{equation}
has an $\FQ$-solution $y_i$. Indeed, as $x \notin \{0,a_i,b_i^{-1}\}$, the term $\frac{x(a_i-x)}{b_ix-1}$ is nonzero. Thus, $y_i \notin \{ 0,a_i-x\}$ and $z_i=a_i-x-y_i \ne 0$.  

By Lemma \ref{lem-quadeqn}(2), Equation \eqref{eqn-quad-odd} has $1+\eta(\Delta_i)$ $\FQ$-solutions, where 
$$
\Delta_i=(x-a_i)^{2}-4\frac{x(a_i-x)}{(b_ix-1)}=\frac{(x-a_i)(b_ix-1)}{(b_ix-1)^2}((x-a_i)(b_ix-1)+4x) 
$$
Thus, Equation~\eqref{eqn-quad-odd} has an $\FQ$-solution if and only if 
$$
\eta((x-a_i)(b_ix-1)((x-a_i)(b_ix-1)+4x)) \ne -1
$$
Set
$$
N=|\big\{x\in\FQ \sm E \mid \eta\big((x-a_i)(b_ix-1)((x-a_i)(b_ix-1)+4x)\big) \ne -1 \text{ for each }1\le i\le t\big\}|.
$$
By Proposition~\ref{prop-oddproof} in Appendix A, we have $N>0$ and therefore, the proof is complete.
\end{proof}

Now we are ready to prove Theorem \ref{thm-asymub}.

\begin{proof}[Proof of Theorem \ref{thm-asymub}]
Note that $m \ge 2(t+1) \log_{q}2+2\log_qt$ is equivalent to $\sqrt{Q} \ge t2^{t+1}$. If $q$ is even, Theorem \ref{thm-asymub} follows from Propositions \ref{prop-even}. If $q$ is odd and $t \ge 2$, Theorem \ref{thm-asymub} follows from Propositions \ref{prop-odd}. If $q$ is odd and $t=1$, then $\sqrt{Q} \ge t2^{t+1}=4$ forces $Q \ge 16$, which excludes the pair $(m,q)=(2,3)$ with $\rho_1(M(2,3))=4$. By Result \ref{res-cr}, $\rho_1(M(m,q)) \le 3$ for each $(m,q) \ne (2,3)$.  
\end{proof}

\section{Generalized covering radius of \texorpdfstring{$M(m,q)$ with $q \ge 4$}{M(m,q) with q >= 4}} \label{sec-qge4}
When $q \ge 4$, the generalized covering radius of $\rho_t(M(m,q))$ can be determined when $m$ is sufficiently large compared with $t$. 

\begin{theorem}
\label{thm-Melasqge4}
Let $q \ge 4$ be a prime power. Let $m \ge 3$ and $3 \le t \le m$. Then 
$$
\rho_t(M(m,q))
\begin{cases}
\in [t+\lc \log_q (t+1) \rc,2t] & \mbox{if $3 \le m < t(2t-1)$ } \\
= 2t & \mbox{if $m \ge t(2t-1)$}
\end{cases}
$$
Moreover, 
$$
\rho_2(M(m,q))=\begin{cases}
                                2 & \mbox{if $m=1$,} \\
                                4 & \mbox{if $m\ge 2$.}
                          \end{cases}
$$ 
\end{theorem}
\begin{proof}
By Result \ref{res-cr}, $\rho_1(M(m,q))=2$ for each $q \ge 4$ and $m \ge 1$. Combining Remark \ref{rem-gcr}(3) and Theorems \ref{thm-GHWlb}(1), \ref{thm-asymlb}, we derive the range of $\rho_t(M(m,q))$ for $t \ge 3$.

For $q \ge 4$ and $m=1$, in view of Remarks \ref{rem-gcr}(2) and \ref{rem-Melas}(1), we have $\rho_2(M(1,q))=2$. Moreover, for $q \ge 4$ and $m \ge 2$, $\rho_2(M(m,q)) \ge 4$ follows from Theorem \ref{thm-gcv2lb}. By Remark \ref{rem-gcr}(3), $\rho_2(M(m,q)) \le 2\rho_1(M(m,q))=4$ Thus, $\rho_2(M(m,q))=4$ for $q \ge 4$ and $m \ge 2$.
\end{proof}

\section{Generalized covering radius of \texorpdfstring{$M(m,3)$}{M(m,3)}} \label{sec-qeq3}
In this section, we consider generalized covering radius of $M(m,3)$. For $q$ being an odd prime power, we use $\square$ (resp. $\boxslash$) to denote the set of all nonzero squares (resp. all nonsquares) in $\Fq$. We use $\eta$ to denote the multiplicative quadratic character over $\Fq$. As a preparation, we have the following lemma concerning character sums of quadratic character.

\begin{lemma}
\label{lem-charsum}
Let $q$ be an odd prime power. Then the following hold true.
\begin{enumerate}[label=(\arabic*)]
\item For each $a \in \Fq^*$, $\sum_{x\in\Fq}\eta(x^2+a)=-1$ and $\sum_{x\in\Fq}\eta(x^2+ax)=-1$.\\
\item $$
\sum_{x\in\Fq}\eta(x^2(x^2-1))=
\begin{cases}
-2 & \mbox{if $q\equiv 1\pmod 4$,} \\
0 & \mbox{if $q\equiv 3\pmod 4$.}
\end{cases}
$$
\item $\sum_{x\in\F_{3^m}}\eta\bigl((x^2+x)(x^2+1)\bigr)=-1-(-1+\sqrt{-2})^m-(-1-\sqrt{-2})^m$.
\item $\sum_{x\in\F_{3^m}}\eta\bigl((x^2-x)(x^2+1)\bigr)=-1-(-1+\sqrt{-2})^m-(-1-\sqrt{-2})^m$.
\item $$
\sum_{x\in\F_{3^m}}\eta(x^2(x^4-1))=
\begin{cases}
0  & \mbox{if $m$ odd,}\\
-2-2(-3)^{\frac{m}{2}} & \mbox{if $m$ even}.
\end{cases}
$$
\end{enumerate}
\end{lemma}

\begin{proof}
\noindent (1) For $a\in\Fq^*$,
\begin{align*}
\sum_{x\in\Fq}\eta(x^2+a)=&\sum_{y\in\Fq}\bigl(1+\eta(y)\bigr)\eta(y+a)=\sum_{y\in\Fq}\eta(y+a)+\sum_{y\in\Fq}\eta(y(y+a))\\
=&\sum_{y\in\Fq^*}\eta(1+\frac{a}{y})=-\eta(1)=-1.
\end{align*}
For $a\in\Fq^*$,
$$
\sum_{x\in\Fq}\eta(x^2+ax)=\sum_{x\in\Fq^*}\eta(x(x+a))=\sum_{x\in\Fq^*}\eta(1+\frac{a}{x})=-\eta(1)=-1.
$$

\noindent (2) By Part (1), we know that $\sum_{x\in\Fq}\eta((x^2-1))=-1$. Moreover,
\begin{align*}
\sum_{x\in\Fq}\eta(x^2(x^2-1))=&\sum_{x\in\Fq^*}\eta(x^2-1)=\sum_{x\in\Fq}\eta(x^2-1)-\eta(-1)=-1-\eta(-1) \\
=&\begin{cases}
-2 & \mbox{if $q\equiv1\pmod 4$,}\\
0 & \mbox{if $q\equiv3\pmod 4$.}
\end{cases}
\end{align*}

\noindent (3) For $x \in \F_{3^m}^*$, substituting $x=\frac{1}{y}$, we have
\begin{align*}
\sum_{x\in\F_{3^m}}\eta\bigl((x^2+x)(x^2+1)\bigr)=&\sum_{x\in\F_{3^m}}\eta(x^4+x^3+x^2+x)=\sum_{y\in\F_{3^m}^*}\eta\left(\frac{1}{y^4}+\frac{1}{y^3}+\frac{1}{y^2}+\frac{1}{y}\right)\\
=&\sum_{y\in\F_{3^m}^*}\eta\left(\frac{y^3+y^2+y+1}{y^4}\right)=\sum_{y\in\F_{3^m}^*}\eta(y^3+y^2+y+1)\\
=&\sum_{y\in\F_{3^m}}\eta(y^3+y^2+y+1)-1.
\end{align*}
Let $N_{3^m}$ be the number of pairs $(x,y) \in \F_{3^m}^2$ satisfying $x^2=y^3+y^2+y+1$. Then
$$
N_{3^m}=\sum_{y\in\F_{3^m}}\bigl(1+\eta(y^3+y^2+y+1)\bigr)
=3^m+\sum_{y\in\F_{3^m}}\eta(y^3+y^2+y+1).
$$
For the elliptic curve $E_1: x^2=y^3+y^2+y+1$, which is nonsingular over $\F_3$, we use $\#E_1(\F_{3^m})$ to denote the number of $\F_{3^m}$-rational points on $E_1$. Therefore, $N_{3^m}=\#E_1(\F_{3^m})-1$ and $\sum_{x\in\F_{3^m}}\eta(x^4+x^3+x^2+x)=\#E_1(\F_{3^m})-3^m-2$. A direct computation shows that $\#E_1(\F_3)=6$. Thus,
$$
a=3+1-\#E_1(\F_3)=-2,
$$
and $\alpha=-1+\sqrt{-2}$, $\beta=-1-\sqrt{-2}$ are the two roots of the polynomial
$$
X^2-aX+3=X^2+2X+3.
$$
By Lemma \ref{lem-elliptic}, we have
$$
\#E_1(\F_{3^m})=3^m+1-(\alpha^m+\beta^m),
$$
which implies
$$
\sum_{x\in\F_{3^m}}\eta(x^4+x^3+x^2+x)=-1-(\alpha^m+\beta^m)=-1-(-1+\sqrt{-2})^m-(-1-\sqrt{-2})^m.
$$

\noindent (4) The proof of Part (4) is analogous to that of Part (3).

\noindent (5) Note that
\begin{align*}
\sum_{x\in\F_{3^m}}\eta(x^2(x^4-1))=&\sum_{x\in\F_{3^m}^*}\eta(x^4-1)=\sum_{x\in\F_{3^m}}\eta(x^4-1)-\eta(-1)\\
=&\sum_{x\in\F_{3^m} \sm \{1\}}\eta(x^4-1)-\eta(-1).
\end{align*}
Given that
\begin{align*}
\sum_{x\in\F_{3^m}\sm \{1\}}\eta(x^4-1)=&\sum_{y\in\F_{3^m}^*}\eta\left(\big(1+\frac{1}{y}\big)^4-1\right)=\sum_{y\in\F_{3^m}^*}\eta\left(\frac{1}{y}+\frac{1}{y^3}+\frac{1}{y^4}\right)\\
=&\sum_{y\in\F_{3^m}^*}\eta\left(\frac{y^3+y+1}{y^4}\right)=\sum_{y\in\F_{3^m}}\eta(y^3+y+1)-1.
\end{align*}
Hence, we have
$$
\sum_{x\in\F_{3^m}}\eta(x^2(x^4-1))=\sum_{y\in\F_{3^m}}\eta(y^3+y+1)-1-\eta(-1)
$$
Let $N_{3^m}$ be the number of pairs $(x,y) \in \F_{3^m}^2$ satisfying $x^2=y^3+y+1$. Then
$$
N_{3^m}=\sum_{y\in\F_{3^m}}\bigl(1+\eta(y^3+y+1)\bigr)=3^m+\sum_{y\in\F_{3^m}}\eta(y^3+y+1).
$$
For the elliptic curve $E_2: x^2=y^3+y+1$, which is nonsingular over $\F_3$, we use $\#E_2(\F_{3^m})$ to denote the number of $\F_{3^m}$-rational points on $E_2$. Therefore, $N_{3^m}=\#E_2(\F_{3^m})-1$ and $\sum_{x\in\F_{3^m}}\eta(x^2(x^4-1))=\#E_2(\F_{3^m})-3^m-2-\eta(-1)$. A direct check shows that $\#E_2(\F_3)=4$. Thus,
$$
a=3+1-\#E_2(\F_3)=0.
$$
The polynomial $T^2-aT+3=T^2+3$ has two roots $\sqrt{-3}$ and $-\sqrt{-3}$. By Lemma \ref{lem-elliptic}, we have
$$
\#E_2(\F_{3^m})=3^m+1-(\sqrt{-3})^m-(-\sqrt{-3})^m
=\begin{cases}
3^m+1, & \mbox{if $m$ odd,}\\
3^m+1-2(-3)^{m/2}, & \mbox{if $m$ even.}
\end{cases}
$$
Consequently,
$$
\sum_{x\in\F_{3^m}}\eta(x^2(x^4-1))
=\begin{cases}
0 & \mbox{if $m$ odd,}\\
-2-2(-3)^{\frac{m}{2}} & \mbox{if $m$ even.}
\end{cases}
$$
\end{proof}

Now we are ready to prove the following auxiliary proposition. 

\begin{proposition}
\label{prop-gamma}
Let $\square$ be the set of all nonzero squares in $\F_{3^m}$. Let $\boxslash$ be the set of all nonsquares in $\F_{3^m}$. Set 
$$
\Gamma_m=\left\{x\in\F_{3^m}\setminus\F_3 \mid x(x+1)\in\boxslash,\ x(x-1)\in\boxslash,\ x^2+1\in\boxslash\right\}.
$$
For each $m\geq 4$, we have $|\Gamma_m| \ge 1$.
\end{proposition}
\begin{proof}
Suppose an element $x \in\F_{3^m}$ satisfies $x(x+1)\in\boxslash$, $x(x-1)\in\boxslash$, and $x^2+1\in\boxslash$. Clearly, we have $x\in\F_{3^m} \sm \F_3$. Define $f: \F_{3^m} \rightarrow \Q$ by
\begin{align*}
f(x)=&\frac{1-\eta(x(x+1))}{2}\cdot\frac{1-\eta(x(x-1))}{2}\cdot\frac{1-\eta(x^2+1)}{2} \\
     =&\frac{1}{8} \big(1-\eta(x(x+1))-\eta(x(x-1))-\eta(x^2+1)+\eta(x^2(x^2-1))+\eta((x^2+x)(x^2+1))\\
       &+\eta((x^2-x)(x^2+1))-\eta(x^2(x^4-1)) \big)
\end{align*}
When none of $x(x+1)$, $x(x-1)$, $x^2+1$ is zero, $f(x) \in \{0,1\}$. Specifically,
$$
f(x)=\begin{cases}
           0 & \mbox{at least one of $x(x+1)$, $x(x-1)$, $x^2+1$ is a nonzero square} \\
           1 & \mbox{all of $x(x+1)$, $x(x-1)$, $x^2+1$ are nonsquares}
       \end{cases}
$$
If at least one of $x(x+1)$, $x(x-1)$, $x^2+1$ is zero, define $\Theta_m=\{ x \in \F_{3^m} \mid x^2+1=0 \}$. Then
$$
\Theta_m=\begin{cases}
                  \es & \mbox{if $m$ odd,} \\
                  \{\theta,-\theta\}, & \mbox{if $m$ even.}
                  \end{cases} 
$$
Note that $x(x+1) \cdot x(x-1) \cdot (x^2+1)=0$ if and only if $x \in \F_3 \cup \Theta_m$. A routine calculation shows
$$
f(x)=\begin{cases}
       0 & \mbox{if $x=0$, $m$ odd; or $x \in \F_3$, $m$ even;} \\
          & \mbox{or one of $\theta+1$ and $\theta-1$ is a nonzero square, $m$ even} \\
       \frac{1}{2} & \mbox{if $x=1,2$, $m$ odd;} \\
                        & \mbox{or $\theta+1$ and $\theta-1$ are nonsquares, $m$ even}
       \end{cases}
$$
Therefore, $\sum_{x \in \F_3 \cup \Theta_m} f(x) \le 1$. Consequently,
\begin{align*}
|\Gamma_m|=&\left|\left\{x\in\F_{3^m}\setminus\F_3 \mid x(x+1)\in\boxslash, x(x-1)\in\boxslash,x^2+1\in\boxslash\right\}\right|\\
                     \ge&\sum_{x\in\F_{3^m}} f(x)- \sum_{x \in \F_3} f(x)- \sum_{x \in \Theta_m} f(x) \\
                     =&\sum_{x\in\F_{3^m}} f(x)- \sum_{x \in \F_3 \cup \Theta_m} f(x) \\
                     \ge& \sum_{x\in\F_{3^m}} f(x)-1
\end{align*}
Using Lemma \ref{lem-charsum}, we have
\begin{align*}
&\sum_{x\in\F_{3^m}}f(x) \\
=&\frac{1}{8}\sum_{x\in\F_{3^m}}\big(1-\eta(x(x+1))-\eta(x(x-1))-\eta(x^2+1)+\eta(x^2(x^2-1))\\
&\qquad\qquad+\eta((x^2+x)(x^2+1))+\eta((x^2-x)(x^2+1))-\eta(x^2(x^4-1))\big)\\
=&\frac{1}{8}\big(3^m-\sum_{x\in\F_{3^m}}\eta(x(x+1))-\sum_{x\in\F_{3^m}}\eta(x(x-1))-\sum_{x\in\F_{3^m}}\eta(x^2+1)+\sum_{x\in\F_{3^m}}\eta(x^2(x^2-1)) \\
&\qquad+\sum_{x\in\F_{3^m}}\eta((x^2+x)(x^2+1))+\sum_{x\in\F_{3^m}}\eta((x^2-x)(x^2+1))-\sum_{x\in\F_{3^m}}\eta(x^2(x^4-1))\big) \\
=&\begin{cases}
\frac18\left(3^m+1-2(-1+\sqrt{-2})^m-2(-1-\sqrt{-2})^m\right), & \mbox{if $m$ odd,}\\
\frac18\left(3^m+1+2(-3)^{\frac{m}{2}}-2(-1+\sqrt{-2})^m-2(-1-\sqrt{-2})^m \right), & \mbox{if $m$ even.}
\end{cases}
\end{align*}
Set $\om_m=(-1+\sqrt{-2})^m$, then $\ol{\om_m}=(-1-\sqrt{-2})^m$. Then $|\om_m+\ol{\om_m}|=2 \cdot 3^{\frac{m}{2}}$. Consequently, together with $m \ge 4$,
$$
|\Gamma_m| \ge \sum_{x\in\F_{3^m}} f(x)-1\ge \frac{1}{8}(3^m-4\cdot3^{m/2}-7)>0
$$
Since $|\Gamma_m|$ is an integer, then $|\Gamma_m| \ge 1$.
\end{proof}

Next, we study a linear system that is closely related to the second generalized covering radius of $M(m,3)$. 

\begin{lemma}
\label{lem-simsys}
Let $\alpha_1, \alpha_2\in\F_3$ and $r_1, r_2 \in \F_{3^m}$ for some $m \ge 1$. The system
$$
\begin{cases}
\alpha_1x_1+\alpha_2x_2=r_1 \\
\alpha_1x_1^{-1}+\alpha_2x_2^{-1}=r_2
\end{cases}
$$
is solvable for $x_1,x_2 \in\F_{3^m}^*$ if and only if one of the following holds true.
\begin{enumerate}[label=(\arabic*)]
\item If $\alpha_1=\alpha_2=0$ and $r_1=r_2=0$, then $x_1,x_2\in\F_{3^m}^*$ are arbitrary.
\item If $\alpha_1=0$, $\alpha_2\ne0$, and $\alpha_2^2=r_1r_2$, then $x_1\in\F_{3^m}^*$ is arbitrary and $x_2=\frac{r_1}{\alpha_2}=\frac{\alpha_2}{r_2}$.
\item If $\alpha_1\ne0$, $\alpha_2=0$, and $\alpha_1^2=r_1r_2$, then $x_1=\frac{r_1}{\alpha_1}=\frac{\alpha_1}{r_2}$ and $x_2\in\F_{3^m}^*$ is arbitrary.
\item If $\alpha_1\alpha_2\ne0$, then the following holds true.
\begin{enumerate}
\item[(4a)] If $r_1=r_2=0$, then $x_1\in\F_{3^m}^*$ is arbitrary and $x_2=-\frac{\alpha_1}{\alpha_2}x_1$.
\item[(4b)] If $r_1,r_2\in\F_{3^m}^*$ and $r_2=r_1^{-1}$, then $r_1r_2(r_1r_2-1)=0$,  $x_1=\frac{r_1}{2\alpha_1}$ and $x_2=\frac{r_1}{2\alpha_2}$.
\item[(4c)] If $r_1,r_2\in\F_{3^m}^*$ with $r_1r_2(r_1r_2-1)\in\square$, set $r_1r_2(r_1r_2-1)=u^2$ for some $u\in\F_{3^m}^*$. Then
$$
x_1=\frac{1}{2}(\frac{r_1}{\alpha_1}\pm\frac{u}{r_2}), \quad x_2=\frac{1}{2}(\frac{r_1}{\alpha_2}\mp\frac{\alpha_1u}{\alpha_2r_2}).
$$
\end{enumerate}
\end{enumerate}
\end{lemma}
\begin{proof}
\noindent (1) If $\alpha_1=\alpha_2=0$, then we must have $r_1=r_2=0$ and $x_1,x_2\in\F_{3^m}^*$ are arbitrary.

\noindent (2) If $\alpha_1=0$ and $\alpha_2\ne0$, then
$$
\begin{cases}
\alpha_2x_2=r_1 \\
\alpha_2x_2^{-1}=r_2
\end{cases}
$$
Thus, we must have $r_1r_2=\alpha_2^2$. For each pair $(r_1,r_2)\in(\F_{3^m}^*)^2$ satisfying $\alpha_2^2=r_1r_2$, we have that $x_1\in\F_{3^m}^*$ is arbitrary and
$x_2=\frac{r_1}{\alpha_2}=\frac{\alpha_2}{r_2}$.

\noindent (3) The proof of Part (3) is analogous to that of Part (2).

\noindent (4) If $\alpha_1\alpha_2\ne0$, then $x_2=\alpha_2^{-1}(r_1-\alpha_1x_1)$. Substituting $x_2$ into the second equation leads to
$$
\alpha_1r_2x_1^2-r_1r_2x_1+\alpha_1r_1=0.
$$
If $r_2=0$, then $r_1=0$ and $x_2=-\frac{\alpha_1}{\alpha_2}x_1$. Hence, $x_1\in\F_{3^m}^*$ is arbitrary and $x_2=-\frac{\alpha_1}{\alpha_2}x_1$, yielding Part (4a).

\noindent If $r_2\ne0$, we have
\begin{equation}
\label{eqn-simsys-quad}
x_1^2-\frac{r_1}{\alpha_1}x_1+\frac{r_1}{r_2}=0.
\end{equation}
In view of Lemma \ref{lem-quadeqn}(2), we compute $\frac{r_1^2}{\al_1^2}-4\frac{r_1}{r_2}=r_1^2-\frac{r_1}{r_2}=\frac{1}{r_2^2}r_1r_2(r_1r_2-1)$. If $r_1=0$, then $\frac{r_1^2}{\al_1^2}-4\frac{r_1}{r_2}=0$, $x_1=\frac{r_1}{2\al_1}=0$, contradicting $x_1 \in \F_{3^m}^*$. If $r_1 \ne 0$, then there are three subcases. If $r_1r_2(r_1r_2-1) \in \boxslash$, then Equation \eqref{eqn-simsys-quad} has no $\F_{3^m}$  solutions. If $r_1r_2(r_1r_2-1)=0$, namely, $r_2=r_1^{-1}$, we have $x_1=\frac{r_1}{2\alpha_1}$ and $x_2=\frac{r_1}{2\alpha_2}$, yielding Part (4b). If $r_1r_2(r_1r_2-1)\in \square$, set $r_1r_2(r_1r_2-1)=u^2$ for some $u \in \F_{3^m}^*$. Then $x_1=\frac{1}{2}(\frac{r_1}{\alpha_1}\pm\frac{u}{r_2})$ and $x_2=\frac{1}{2}(\frac{r_1}{\alpha_2}\mp\frac{\alpha_1u}{\alpha_2r_2})$ with $x_1x_2=\frac{r_1}{r_2} \ne 0$, yielding Part (4c). 
\end{proof}

As an application of Lemma \ref{lem-simsys}, we have the following corollary.

\begin{corollary}
\label{cor-simsys}
For $1 \le i \le 4$, let $\alpha_i, \beta_i \in \F_3$.
\begin{enumerate}[label=(\arabic*)]
\item For $\theta \in \F_{3^m}^*$, the system 
$$
\begin{cases}
\alpha_1x_1+\alpha_2x_2+\alpha_3x_3+\alpha_4x_4=0 \\
\alpha_1x_1^{-1}+\alpha_2x_2^{-1}+\alpha_3x_3^{-1}+\alpha_4x_4^{-1}=\theta
\end{cases}
$$
is unsolvable for $x_i \in \F_{3^m}^*$ if at least two of $\al_i$'s are zero.
\item For $\theta \in \F_{3^m}^*$, the system 
$$
\begin{cases}
\beta_1x_1+\beta_2x_2+\beta_3x_3+\beta_4x_4=\theta \\
\beta_1x_1^{-1}+\beta_2x_2^{-1}+\beta_3x_3^{-1}+\beta_4x_4^{-1}=0
\end{cases}
$$
is unsolvable for $x_i \in \F_{3^m}^*$ if at least two of $\be_i$'s are zero.
\item For $m \ge 4$, let $\ga \in \F_{3^m} \sm \F_3$, such that $\ga(\ga+1) \in \boxslash$ and $\ga(\ga-1) \in \boxslash$. Consider the following two systems
$$
\begin{cases}
\alpha_1x_1+\alpha_2x_2+\alpha_3x_3+\alpha_4x_4=0 \\
\alpha_1x_1^{-1}+\alpha_2x_2^{-1}+\alpha_3x_3^{-1}+\alpha_4x_4^{-1}=1
\end{cases}
$$
and
$$
\begin{cases}
\beta_1x_1+\beta_2x_2+\beta_3x_3+\beta_4x_4=\gamma \\
\beta_1x_1^{-1}+\beta_2x_2^{-1}+\beta_3x_3^{-1}+\beta_4x_4^{-1}=0
\end{cases}
$$
They are simultaneously solvable for $x_i \in \F_{3^m}^*$ only if the following hold true
\begin{enumerate}
\item[(3a)] Exactly one of $\al_i$, $1 \le i \le 4$, is zero and exactly one of $\be_i$, $1 \le i \le 4$, is zero.
\item[(3b)] Suppose $\al_i=\be_j=0$ for $1 \le i,j \le 4$, then $i \ne j$. 
\end{enumerate}
\end{enumerate}
\end{corollary}
\begin{proof}
(1) Without loss of generality, we can assume $\al_1=\al_2=0$. Then the system reduces to
$$
\begin{cases}
\alpha_3x_3+\alpha_4x_4=0 \\
\alpha_3x_3^{-1}+\alpha_4x_4^{-1}=\theta
\end{cases}
$$
which is unsolvable for $x_3, x_4 \in \F_{3^m}^*$ by Lemma \ref{lem-simsys}.

\noindent (2) The proof of Part (2) is analogous to that of Part (1).

\noindent (3) For $m \ge 4$, the existence of $\ga \in \F_{3^m} \sm \F_3$, such that $\ga(\ga+1) \in \boxslash$ and $\ga(\ga-1) \in \boxslash$, follows from Proposition \ref{prop-gamma}. Combining Parts (1) and (2), we know that the two systems are simultaneously solvable only if at most one of $\al_i$, $1 \le i \le 4$, is zero and at most one of $\be_i$, $1 \le i \le 4$, is zero. Without loss of generality, it suffices to establish the two systems are not simultaneously solvable in the following four subcases.

\noindent The first subcase concerns $\al_1\al_2\al_3\al_4\be_1\be_2\be_3\be_4 \ne 0$. Combining the two systems, we obtain
$$
\begin{cases}
(\beta_1 \pm \alpha_1)x_1+(\beta_2 \pm \alpha_2)x_2+(\beta_3 \pm \alpha_3)x_3+(\beta_4 \pm \alpha_4)x_4=\gamma \\
(\beta_1 \pm \alpha_1)x_1^{-1}+(\beta_2 \pm \alpha_2)x_2^{-1}+(\beta_3 \pm \alpha_3)x_3^{-1}+(\beta_4 \pm \alpha_4)x_4^{-1}=\pm 1
\end{cases}
$$
Note that $\al_i, \be_i \in \F_3^*$ for each $1 \le i \le 4$. Among the four pairs $(\al_i,\beta_i)$, $1 \le i \le 4$, if $k$ pairs satisfy $\be_i=\al_i$, then $4-k$ pairs satisfy $\be_i=-\al_i$. Therefore, by choosing between the plus sign and the minus sign appropriately, we can make the multiset $\{\{ \beta_1 \pm \alpha_1, \beta_2 \pm \alpha_2, \beta_3 \pm \alpha_3, \beta_4 \pm \alpha_4 \}\}$ contain either $2$ or $3$ or $4$ zeroes. If the multiset contains exactly $4$ zeroes, this contradicts $(\beta_1 \pm \alpha_1)x_1+(\beta_2 \pm \alpha_2)x_2+(\beta_3 \pm \alpha_3)x_3+(\beta_4 \pm \alpha_4)x_4=\gamma$. 
If the multiset contains exactly $3$ zeroes, without loss of generality, assume $\beta_1 \pm \alpha_1 \ne 0$ and $\beta_2 \pm \alpha_2=\beta_3 \pm \alpha_3=\beta_4 \pm \alpha_4=0$. Then the system reduces to 
$$
\begin{cases}
(\beta_1 \pm \alpha_1)x_1=\gamma \\
(\beta_1 \pm \alpha_1)x_1^{-1}=\pm 1
\end{cases}
$$
which leads to $\pm \gamma=1$, contradicting $\ga \in \F_{3^m} \sm \F_3$. If the multiset contains exactly $2$ zeroes, without loss of generality, assume $\beta_1 \pm \alpha_1 \ne 0$, $\beta_2 \pm \alpha_2 \ne 0$ and $\beta_3 \pm \alpha_3=\beta_4 \pm \alpha_4=0$. Then the system reduces to 
$$
\begin{cases}
(\beta_1 \pm \alpha_1)x_1+(\beta_2 \pm \alpha_2)x_2=\gamma \\
(\beta_1 \pm \alpha_1)x_1^{-1}+(\beta_2 \pm \alpha_2)x_2^{-1}=\pm 1
\end{cases}
$$
By Lemma \ref{lem-simsys}, since $\ga(\ga+1) \in \boxslash$ and $\ga(\ga-1) \in \boxslash$, the above system is unsolvable for $x_1, x_2 \in \F_{3^m}^*$.

\noindent The second subcase concerns $\al_1=0$ and $\al_2\al_3\al_4\be_1\be_2\be_3\be_4 \ne 0$. The third subcase concerns $\be_1=0$ and $\al_1\al_2\al_3\al_4\be_2\be_3\be_4 \ne 0$. Both of them can be proved similarly as the first subcase.

\noindent The fourth subcase concerns $\al_1=\be_1=0$ and $\al_2\al_3\al_4\be_2\be_3\be_4 \ne 0$. Combining the two systems, we obtain
$$
\begin{cases}
(\beta_2 \pm \alpha_2)x_2+(\beta_3 \pm \alpha_3)x_3+(\beta_4 \pm \alpha_4)x_4=\gamma \\
(\beta_2 \pm \alpha_2)x_2^{-1}+(\beta_3 \pm \alpha_3)x_3^{-1}+(\beta_4 \pm \alpha_4)x_4^{-1}=\pm 1
\end{cases}
$$
Likewise, by choosing between the plus sign and minus sign appropriately, we can make the multiset $\{\{ \beta_2 \pm \alpha_2, \beta_3 \pm \alpha_3, \beta_4 \pm \alpha_4 \}\}$ contain either $2$ or $3$ zeroes. In both cases, the above system is unsolvable for $x_i \in \F_{3^m}^*$.

After excluding the above four subcases, we have exactly two subcases corresponding to (3a) and (3b) left. 
\end{proof}

The following proposition leads to a lower bound on $\rho_2(M(m,3))$.

\begin{proposition}
\label{prop-q3}
Let $m\geq4$. Let $\gamma\in\F_{3^m}\setminus\F_3$ such that
$$
\begin{cases}
\gamma(\gamma+1) \in\boxslash \\
\gamma(\gamma-1)  \in\boxslash \\
\gamma^2+1\in\boxslash
\end{cases}
$$
Then the following two systems are not simultaneously solvable for distinct $x_i\in\F_{3^m}^*$, $1\leq i\leq4$,
$$
\begin{cases}
\alpha_1x_1+\alpha_2x_2+\alpha_3x_3+\alpha_4x_4=0 \\
\alpha_1x_1^{-1}+\alpha_2x_2^{-1}+\alpha_3x_3^{-1}+\alpha_4x_4^{-1}=1
\end{cases}
$$
and
$$
\begin{cases}
\beta_1x_1+\beta_2x_2+\beta_3x_3+\beta_4x_4=\gamma \\
\beta_1x_1^{-1}+\beta_2x_2^{-1}+\beta_3x_3^{-1}+\beta_4x_4^{-1}=0
\end{cases}
$$
where $\alpha_i,\beta_i\in\F_3$ for each $1\leq i\leq4$. In particular, for $m \ge 4$, we have $\rho_2(M(m,3)) \ge 5$.
\end{proposition} 
\begin{proof}
For $m \ge 4$, the existence of $\ga \in \F_{3^m} \sm \F_3$ satisfying $\gamma(\gamma+1) \in\boxslash$, $\gamma(\gamma-1) \in\boxslash$, and $\gamma^2+1\in\boxslash$ follows from Proposition \ref{prop-gamma}. Note that if $\rho_2(M(m,3)) \le 4$, there exist distinct $x_i \in \F_{3^m}^*$, $1 \le i \le 4$, such that the two systems are simultaneously solvable, where $\al_i, \be_i \in \F_3$ are the coefficients forming $\F_3$-linear combination of at most four columns within the parity-check matrix of $M(m,3)$. Hence, in order to show $\rho_2(M(m,3)) \ge 5$ for $m \ge 4$, it suffices to show that the two systems are not simultaneously solvable whenever $m \ge 4$.

\noindent In view of Corollary \ref{cor-simsys}(3), the two systems are simultaneously solvable for $x_i \in \F_{3^m}^*$ only if exactly one of $\al_i$, $1 \le i \le 4$, is zero and exactly one of $\be_i$, $1 \le i \le 4$, is zero.

\noindent (1) If $\alpha_i=\beta_i$ for each $1\leq i \leq4$, then $\ga=0$, contradicting $\ga \in \F_{3^m} \sm \F_3$.

\noindent (2) If $\alpha_i=\beta_i$ for exactly three distinct $i$'s, without loss of generality, we can assume $\alpha_i=\beta_i$ for exactly $1\leq i\leq3$, then
$$
\begin{cases}
(\alpha_4-\beta_4)x_4=-\ga \\
(\alpha_4-\beta_4)x_4^{-1}=1
\end{cases}
$$
which implies $-\ga=(\alpha_4-\beta_4)^2 \in\F_3$, contradicting $\ga \in \F_{3^m} \sm \F_3$.

\noindent (3) If $\alpha_i=\beta_i$ for exactly two distinct $i$'s, without loss of generality, we can assume $\alpha_i=\beta_i$ for exactly $1\leq i\leq2$, then
$$
\begin{cases}
(\alpha_3-\beta_3)x_3+(\alpha_4-\beta_4)x_4=-\gamma,\\
(\alpha_3-\beta_3)x_3^{-1}+(\alpha_4-\beta_4)x_4^{-1}=1.
\end{cases}
$$
Since $-\gamma(-\gamma-1)=\gamma(\gamma+1)\in\boxslash$, by Lemma \ref{lem-simsys}, the above system is unsolvable for $x_3,x_4\in\F_{3^m}^*$.

\noindent (4) If $\alpha_i=\beta_i$ for exactly one $i$, without loss of generality, we can assume $\alpha_1=\beta_1$, and therefore,
$$
\begin{cases}
\alpha_1x_1+\alpha_2x_2+\alpha_3x_3+\alpha_4x_4=0 \\
\alpha_1x_1^{-1}+\alpha_2x_2^{-1}+\alpha_3x_3^{-1}+\alpha_4x_4^{-1}=1
\end{cases}
$$
and
$$
\begin{cases}
\alpha_1x_1+\beta_2x_2+\beta_3x_3+\beta_4x_4=\gamma \\
\alpha_1x_1^{-1}+\beta_2x_2^{-1}+\beta_3x_3^{-1}+\beta_4x_4^{-1}=0
\end{cases}
$$
In view of Corollary \ref{cor-simsys}(3), the two above systems are simultaneously solvable for $x_i \in \F_{3^m}^*$ only if exactly one of $\al_i$, $2 \le i \le 4$, is zero and exactly one of $\be_i$, $2 \le i \le 4$, is zero. Moreover, if $\al_i=\be_j=0$ for some $2 \le i,j \le 4$, then $i \ne j$. Without loss of generality, we can assume $\alpha_2=\beta_3=0$ and $\beta_4=2\alpha_4$. Therefore, we can rephrase the systems into the following two:
\begin{equation}
\label{eqn-sysA}
\begin{cases}
\alpha_1x_1+\alpha_3x_3+\alpha_4x_4=0\\
\alpha_1x_1^{-1}+\alpha_3x_3^{-1}+\alpha_4x_4^{-1}=1
\end{cases}
\end{equation}
and
\begin{equation}
\label{eqn-sysB}
\begin{cases}
2\alpha_1x_1+2\beta_2x_2+\alpha_4x_4=-\gamma,\\
2\alpha_1x_1^{-1}+2\beta_2x_2^{-1}+\alpha_4x_4^{-1}=0.
\end{cases}
\end{equation}
The system \eqref{eqn-sysA} can be rewritten as
\begin{numcases}{}
  x_3 = \frac{-\alpha_1 x_1 - \alpha_4 x_4}{\alpha_3} \nonumber \\
  \alpha_1 x_1^2 x_4 + \alpha_4 x_1 x_4^2 - \alpha_1\alpha_4(x_1^2+x_4^2) - x_1 x_4 = 0 \label{eqn-A2}
\end{numcases}
If $x_3=0$, then $x_1=-\frac{\alpha_4}{\alpha_1}x_4$. Substituting $x_1$ into \eqref{eqn-A2} gives $2\alpha_1\alpha_4x_4^2=0$, which implies $x_1=x_4=0$. Hence, 
each $(x_1,x_4)$ solution to \eqref{eqn-A2} with $x_1x_4\ne0$ leads to a solution $(x_1,x_3,x_4)$ to \eqref{eqn-sysA} with $x_1x_3x_4\ne0$.

\noindent The system \eqref{eqn-sysB} can be rewritten as
\begin{numcases}{}
  x_2=\frac{\alpha_1x_1-\alpha_4x_4-\gamma}{2\beta_2} \nonumber \\
  \alpha_1\alpha_4(x_1^2+x_4^2)+\gamma(\alpha_1x_4-\alpha_4x_1)-x_1x_4=0 \label{eqn-B2}
\end{numcases}
If $x_2=0$, then $x_1=\frac{\al_4x_4+\ga}{\al_1}$. Substituting $x_1$ into \eqref{eqn-B2} gives $x_4=-\al_4\ga$ and $x_1=0$. Hence, each $(x_1,x_4)$ solution to \eqref{eqn-B2} with $x_1x_4\ne0$ leads to a solution $(x_1,x_2,x_4)$ to \eqref{eqn-sysB} with $x_1x_2x_4\ne0$.

\noindent Consequently, in order to show that systems \eqref{eqn-sysA} and \eqref{eqn-sysB} are not simultaneously solvable for $x_1,x_2,x_3,x_4\in\F_{3^m}^*$, it suffices to show that the following system formed by \eqref{eqn-A2} and \eqref{eqn-B2} is not solvable for $x_1,x_4\in\F_{3^m}^*$:
$$
\begin{cases}
\alpha_1 x_1^2 x_4 + \alpha_4 x_1 x_4^2 - \alpha_1\alpha_4(x_1^2+x_4^2) - x_1 x_4 = 0 \\
\alpha_1\alpha_4(x_1^2+x_4^2)+\gamma(\alpha_1x_4-\alpha_4x_1)-x_1x_4=0
\end{cases}
$$
If $\al_1=0$, then the second equation of the above system simplifies to $\ga\al_4^2=-1$, which contradicts $\ga \in \F_{3^m} \sm \F_3$. Thus, we have $\alpha_1\alpha_4x_1x_4 \ne 0$. By dividing $\alpha_1\alpha_4x_1x_4$ on both sides, the system can be rephrased as
\begin{equation}
\label{eqn-sysC}
\begin{cases}
 \frac{x_1}{x_4}+\frac{x_4}{x_1}+\frac{1}{\alpha_1\alpha_4}=\frac{x_1}{\al_4}+\frac{x_4}{\al_1} \\
 \frac{x_1}{x_4}+\frac{x_4}{x_1}-\frac{1}{\alpha_1\alpha_4}=-\ga(\frac{1}{\al_4x_1}-\frac{1}{\al_1x_4})
\end{cases}
\end{equation}
Multiplying the left-hand side of \eqref{eqn-sysC}, we have
\begin{equation}
\label{eqn-sysCL}
\begin{split}
\left(\frac{x_1}{x_4}+\frac{x_4}{x_1}+\frac{1}{\alpha_1\alpha_4}\right)\left(\frac{x_1}{x_4}+\frac{x_4}{x_1}-\frac{1}{\alpha_1\alpha_4}\right)=&\left(\frac{x_1}{x_4}+\frac{x_4}{x_1}+1\right)\left(\frac{x_1}{x_4}+\frac{x_4}{x_1}-1\right) \\
=&\left( \frac{x_1}{x_4}-\frac{x_4}{x_1}\right)^2
\end{split}
\end{equation}
Multiplying the right-hand side of \eqref{eqn-sysC}, we have
\begin{equation}
\label{eqn-sysCR}
-\gamma\left(\frac{x_1}{\alpha_4}+\frac{x_4}{\alpha_1}\right)\left(\frac{1}{\alpha_4x_1}-\frac{1}{\alpha_1x_4}\right)=\frac{\ga}{\al_1\al_4}\left(\frac{x_1}{x_4}-\frac{x_4}{x_1}\right)
\end{equation}
In order to show that the system \eqref{eqn-sysC} is unsolvable for $x_1, x_4 \in \F_{3^m}^*$, it boils down to examine when the right hand sides of Equations \eqref{eqn-sysCL} and \eqref{eqn-sysCR} agree.

\noindent If $\frac{x_1}{x_4}=\frac{x_4}{x_1}$, then $\frac{x_1}{x_4}=\pm1$. Since $x_1$ and $x_4$ are distinct, we have $x_4=-x_1$. The system \eqref{eqn-sysC} simplifies to
$$
\begin{cases}
 (\frac{1}{\al_4}-\frac{1}{\al_1})x_1=\frac{1}{\al_1\al_4}-2 \\
 -\ga(\frac{1}{\al_1}+\frac{1}{\al_4})x_1^{-1}=-\frac{1}{\al_1\al_4}-2
\end{cases}
$$
Recall that $\al_1\al_4 \ne 0$. In both cases of $\al_1=\al_4$ and $\al_1 \ne \al_4$, it is straightforward to check the above system is unsolvable for $x_1, x_4 \in \F_{3^m}^*$.
\noindent If $\frac{x_1}{x_4} \ne \frac{x_4}{x_1}$, equating the right hand sides of Equations \eqref{eqn-sysCL} and \eqref{eqn-sysCR} leads to
$$
\frac{x_1}{x_4}-\frac{x_4}{x_1}=\frac{\gamma}{\alpha_1\alpha_4}
$$
which is equivalent to
$$
\left(\frac{x_1}{x_4}\right)^2-\frac{\gamma}{\alpha_1\alpha_4}\frac{x_1}{x_4}-1=0.
$$
In view of Lemma \ref{lem-quadeqn}(2), note that $(\frac{\gamma}{\alpha_1\alpha_4})^2+4=\gamma^2+1\in\boxslash$.
Then there exists no $x_1, x_4\in\F_{3^m}^*$ such that the right hand sides of Equations \eqref{eqn-sysCL} and \eqref{eqn-sysCR} agree. Consequently, the system \eqref{eqn-sysC} is unsolvable for $x_1, x_4\in\F_{3^m}^*$.

\noindent (5) If $\alpha_i\ne\beta_i$ for each $1\leq i\leq4$, in view of Corollary \ref{cor-simsys}(3), the two systems are simultaneously solvable for $x_i \in \F_{3^m}^*$ only if exactly one of $\al_i$, $1 \le i \le 4$, is zero and exactly one of $\be_i$, $1 \le i \le 4$, is zero. Moreover, if $\al_i=\be_j=0$ for some $1 \le i,j \le 4$, then $i \ne j$. Without loss of generality, we can assume $\alpha_1=\beta_2=0$,  $\beta_3=2\alpha_3$, $\beta_4=2\alpha_4$. Therefore,

$$
\begin{cases}
\alpha_2x_2+\alpha_3x_3+\alpha_4x_4=0 \\
\alpha_2x_2^{-1}+\alpha_3x_3^{-1}+\alpha_4x_4^{-1}=1
\end{cases}
$$
and
$$
\begin{cases}
\be_1x_1+2\al_3x_3+2\al_4x_4=\gamma \\
\be_1x_1^{-1}+2\al_3x_3^{-1}+2\al_4x_4^{-1}=0
\end{cases}
$$
Consequently,
$$
\begin{cases}
\beta_1x_1+\alpha_2x_2=\gamma,\\
\beta_1x_1^{-1}+\alpha_2x_2^{-1}=1.
\end{cases}
$$
Since $\gamma(\gamma-1)\in\boxslash$, by Lemma \ref{lem-simsys}, the above system is unsolvable for $x_1,x_2\in\F_{3^m}^*$.
\end{proof}

To sum up, we have the following result.

\begin{theorem}
\label{thm-Melasqeq3}
Let $m \ge 3$ and $3 \le t \le m$. Then 
$$
\rho_t(M(m,3)) \in 
\begin{cases}
[t+\lc \log_3 (2t+3) \rc,\min\{3t,2m\}] & \mbox{if $3 \le m < 2(t+1)\log_3 2+2\log_3 t$ } \\
[t+\lc \log_3 (2t+3) \rc,2t+1] & \mbox{if $2(t+1)\log_3 2+2\log_3 t \le m < t(2t-1)$} \\
[2t,2t+1] & \mbox{if $m \ge t(2t-1)$}
\end{cases}
$$
Moreover,
$$
\rho_2(M(m,3))=
\begin{cases}
                     1 & \mbox{if $m=1$} \\
                     4  & \mbox{if $m=2$} \\
                     5 & \mbox{if $m \ge 3$} 
\end{cases}
$$
\end{theorem}
\begin{proof}
By Result \ref{res-cr}, $\rho_1(M(m,3))=3$ for $m \ge 3$. Combining Remark \ref{rem-gcr}(3) and Theorems \ref{thm-GHWlb}(2), \ref{thm-asymlb}, \ref{thm-asymub}, we derive the range of $\rho_t(M(m,3))$ for $t \ge 3$.

Moreover, combining Theorem \ref{thm-asymub} and Proposition \ref{prop-q3}, we have $\rho_2(M(m,3))=5$ for $m \ge 6$. In view of Remark \ref{rem-Melas}(1), we have $\rho_2(M(1,3))=1$. A numeric experiment indicates that $\rho_2(M(2,3))=4$ and $\rho_2(M(3,3))=\rho_2(M(4,3))=\rho_2(M(5,3))=5$.
\end{proof}

\section{Generalized covering radius of \texorpdfstring{$M(m,2)$}{M(m,2)}} \label{sec-qeq2}
In this section, we consider generalized covering radius of $M(m,2)$. We first have the following result concerning the second generalized covering radius.

\begin{theorem}
\label{thm-Melasqeq2}
Let $m \ge 3$ and $3 \le t \le m$. Then 
$$
\rho_t(M(m,2)) \in 
\begin{cases}
[t+\lc \log_2 (t+1) \rc,\min\{3t,2m\}] & \mbox{if $3 \le m < 2(t+1)+2\log_2 t$ } \\
[t+\lc \log_2 (t+1) \rc,2t+1] & \mbox{if $2(t+1)+2\log_2 t \le m < t(2t-1)$} \\
[2t,2t+1] & \mbox{if $m \ge t(2t-1)$}
\end{cases}
$$
Moreover,
$$
\rho_2(M(m,2))=
\begin{cases}
                     1 & \mbox{if $m=1$} \\
                     2 & \mbox{if $m=2$} \\
                     5 & \mbox{if $m=3$} \\
                     6 & \mbox{if $m=4$} \\ 
                     5 & \mbox{if $m \ge 5$} 
\end{cases}
$$
\end{theorem}
\begin{proof}
By Result \ref{res-cr}, $\rho_1(M(m,2))=3$ for $m \ge 3$. Combining Remark \ref{rem-gcr}(3) and Theorems \ref{thm-GHWlb}(1), \ref{thm-asymlb}, \ref{thm-asymub}, we derive the range of $\rho_t(M(m,2))$ for $t \ge 3$.

Moreover, combining Propositions~\ref{prop-supercode}, \ref{prop-GHW}(1), Theorem \ref{thm-asymub}, we have $\rho_2(M(m,2))=5$ for $m \ge 8$. In view of Remark \ref{rem-Melas}(1), we have $\rho_2(M(1,2))=1$. A numeric experiment indicates that $\rho_2(M(2,2))=2$,  $\rho_2(M(3,2))=5$, $\rho_2(M(4,2))=6$, and $\rho_2(M(5,2))=\rho_2(M(6,2))=\rho_2(M(7,2))=5$.
\end{proof}

Combining Theorems \ref{thm-Melasqge4}, \ref{thm-Melasqeq3}, \ref{thm-Melasqeq2}, we establish Theorems \ref{thm-seccr}, \ref{thm-asymcr}.

\section{Conclusion}
\label{sec-conclusion}

In this paper, we investigate the generalized covering radii of Melas codes. For a Melas code $M(m,q)$, we establish lower and upper bounds on its generalized covering radii. Applying these bounds, for $3 \le t \le m$, whenever $m \ge t(2t-1)$, we have
$$
\rho_t(M(m,q))  \begin{cases}
                               \in  [2t,2t+1] & \mbox{if $q \in \{2,3\}$} \\
                                =   2t & \mbox{if $q \ge 4$}
                             \end{cases}
$$
Moreover, we completely determine the second generalized covering radius of Melas codes.

\section*{Appendix A}
\label{sec-appendix}

In this appendix, we establish two propositions that have been used in the proof of Propositions \ref{prop-even} and \ref{prop-odd}.

\begin{proposition}\label{prop-evenproof}
Let $q$ be even, $1 \le t \le m$, and $Q=q^m$. Suppose $\sqrt Q \ge t2^{t+1}$. Let $A$ and $B$ be two subsets of $\FQ$, such that 
$$
A=\{a_i \mid 1\le i\le t\},\qquad B=\{b_i^{-1} \mid 1\le i\le t\},
$$
where $b_i \ne 0$ for each $1 \le i \le t$. Set $E=\{0\} \cup A \cup B$. Define 
$$
N=|\{x \in\FQ\sm E \mid \Tr\Big(\frac{x}{(b_ix+1)(x+a_i)}\Big)=0 \mbox{ for each $1\le i \le t$} \}|.
$$
Then $N>0$.
\end{proposition}
\begin{proof}
For $1 \le i \le t$, define $R_i(X)=\frac{X}{(b_iX+1)(X+a_i)} \in \FQ(X)$ and for $\es \ne I \subseteq[t]$, set $R_I=\sum_{i\in I}R_i$. We have
\begin{equation}
\label{eqn-Neven}
\begin{split}
2^{t}N&=\sum_{x\in\FQ\sm E} \prod_{i=1}^{t}\Big(1+\varphi\big(R_i(x)\big)\Big) \\
          &=\sum_{x\in\FQ\sm E} \Big(1+\sum_{\es \subsetneq I\subseteq[t]} \prod_{i\in I}\varphi\big(R_i(x)\big)\Big) \\
          &=Q-|E|+\sum_{x\in\FQ\sm E} \sum_{\es \subsetneq I\subseteq[t]} \varphi\big(\sum_{i\in I}R_i(x)\big) \\
          &=Q-|E|+\sum_{\es \subsetneq I\subseteq[t]}\sum_{x\in\FQ\sm E} \varphi\big(R_I(x)\big)
\end{split}
\end{equation}
Next, we focus on the estimate of $\sum_{x\in\FQ\sm E} \varphi\big(R_I(x)\big)$. As $b_i\ne0$, $R_i$ has its poles among $\{a_i,b_i^{-1}\}$. For $\es \ne I \subseteq [t]$, let $S_I$ be the set of poles of $R_I$, which satisfies
\begin{equation*}
S_I \subseteq\ A\cup B\ \subset E
\end{equation*}
Thus, every $R_I$ is defined at every $x\in\FQ\sm E$. Set
$$
\cR=\big\{ h^{2}+h+c\ \mid h\in\FQ(X),\ c \in \FQ\,\big\} \subseteq \FQ(X)
$$

If $R_I \in \cR$, then $R_I=h^{2}+h+c$ with $h\in\FQ(X)$ and $c\in\FQ$. For each $\es \ne I \subseteq [t]$, as $R_I(\infty)=0$, $\infty$ is not a pole of $R_I$ and $\infty$ is not a pole of $h$. Evaluating $R_I$ at $\infty$, we have $c=R_I(\infty)+h^2(\infty)+h(\infty)=h^2(\infty)+h(\infty)$. Therefore, $\Tr(c)=0$. Note that every pole of $h$ is a pole of $R_I$, hence $h$ is defined on $\FQ \sm E$. Consequently,
\begin{equation}
\label{eqn-RinR}
 \sum_{x\in\FQ\sm E}\varphi\big(R_I(x)\big)=\sum_{x\in\FQ\sm E}\varphi\big(h^{2}(x)+h(x)+c\big)=\sum_{x\in\FQ\sm E}\varphi\big(c\big)=Q-|E| \ge 0
\end{equation}

If $R_I \notin \cR$, by the definition of $R_I$, we can write
$$
R_I=c_0+\frac{v(X)}{w(X)},\qquad c_0\in\FQ,\quad \deg v<\deg w,\quad \gcd(v,w)=1,
$$
where $w(X)=\prod_{j=1}^{k}w_j(X)^{\,m_j}$ with $w_1(X),\dots,w_k(X)\in\FQ[X]$ being pairwise distinct, monic, and irreducible over $\FQ$. By the definition of $R_i$, we have $\deg w=\sum_j m_j\deg w_j\le 2|I|$. Applying Lemma~\ref{lem-Weilrational}(2), there exists $h_I \in \FQ(x)$, such that $\tilde R_I=R_I-(h_I^2+h_I)$, where $\tilde R_I$ has the set of poles $\tilde{S_I} \subseteq S_I$ and $\varphi(R_I(x))=\varphi(\tilde{R_I}(x))$ for each $x \in (\FQ \sm E) \subset (\FQ \sm S_I)$. Moreover,  
$$
\td{R_I}=\td{c_0}+\frac{\td{v}(X)}{\td{w}(X)}, \qquad \td{c_0} \in\FQ, \quad \deg \td{v}<\deg \td{w},\quad \gcd(\td{v},\td{w})=1,
$$
where $\tilde w(X)=\prod_{j=1}^{k}w_j(X)^{\tilde m_j}$ with $\td{m_j} \le m_j$ being zero or odd for each $1 \le j \le k$. Note that $R_I \notin \cR$ guarantees $\td{R_I}$ is nonconstant. Applying Lemma~\ref{lem-Weilrational}(1) to $\tilde R_I$, we have
$$
\Big|\sum_{x\in \FQ \sm \tilde S_I} \varphi\big(\tilde R_I(x)\big)\Big|\le 1+(\td{L_I}-2)\sqrt Q,
$$
where $\td{L_I}=\sum_j(\tilde m_j+1)\deg w_j \le \sum_j m_j\deg w_j+\sum_j \deg w_j \le 2|I|+2|I|=4|I|$. Note that $|E| \le 2t+1$, we have
\begin{equation}
\label{eqn-RnotinR}
\begin{split}
\Big|\sum_{x\in\FQ\sm E}\varphi\big(R_I(x)\big)\Big|=&\Big|\sum_{x\in\FQ\sm\tilde S_I}\varphi\big(\tilde R_I(x)\big)-\sum_{x\in E\sm\tilde S_I}\varphi\big(\tilde R_I(x)\big)\Big| \\
                                                                                 \le&\Big|\sum_{x\in\FQ\sm\tilde S_I}\varphi\big(\tilde R_I(x)\big)\Big|+\Big|\sum_{x\in E\sm\tilde S_I}\varphi\big(\tilde R_I(x)\big)\Big| \\
                                                                                 \le& 1+(\td{L_I}-2)\sqrt Q+|E\sm\tilde S_I| \\
                                                                                 \le& 1+(4|I|-2)\sqrt Q+2t+1 \\
                                                                                 =& (4|I|-2)\sqrt Q+2t+2
\end{split}                                                                                 
\end{equation}

Combining Equations \eqref{eqn-Neven}, \eqref{eqn-RinR}, \eqref{eqn-RnotinR}, together with $\sum_{\es\ne I\subseteq[t]}|I|=\sum_{i=1}^{t}\binom{t}{i}i=t2^{t-1}$ and $\sqrt{Q} \ge t2^{t+1}$, we have
\begin{align*}
2^{t}N&=Q-|E|+\sum_{\substack{\es \subsetneq I\subseteq[t] \\ R_I \in \cR} } \sum_{x\in\FQ\sm E} \varphi\big(R_I(x)\big)+\sum_{\substack{\es \subsetneq I\subseteq[t] \\ R_I \not\in \cR} } \sum_{x\in\FQ\sm E} \varphi\big(R_I(x)\big) \\
          &\ge Q-2t-1-\sum_{\es \subsetneq I\subseteq[t] } \big((4|I|-2)\sqrt Q+2t+2\big) \\
          &= Q-2t-1-\Big(4\,t2^{t-1}-2(2^{t}-1)\Big)\sqrt Q-(2t+2)(2^{t}-1) \\
          &=Q-t2^{t+1}\sqrt Q+2(2^{t}-1)\sqrt Q-(2t+2)(2^{t}-1)-(2t+1)\\
          &=\sqrt Q\Big(\sqrt Q-t2^{t+1}\Big)+(2^{t}-1)\Big(2\sqrt Q-2t-2\Big)-(2t+1) \\
          &\ge (2^{t}-1)\Big(2^{t+2}t-2t-2\Big)-(2t+1) \\
          &>0
\end{align*}
Consequently, $N>0$. 
\end{proof}

\begin{proposition}\label{prop-oddproof}
Let $q$ be odd, $2 \le t \le m$, and $Q=q^m$. Suppose $\sqrt Q \ge t2^{t+1}$. Let $A$ and $B$ be two subsets of $\FQ$, such that 
$$
A=\{a_i \mid 1\le i\le t\},\qquad B=\{b_i^{-1} \mid 1\le i\le t\},
$$
where $b_i \ne 0$ for each $1 \le i \le t$. Set $E=\{0\} \cup A \cup B$. Define 
$$
N=|\big\{x\in\FQ \sm E \mid \eta\big((x-a_i)(b_ix-1)((x-a_i)(b_ix-1)+4x)\big) \ne -1 \text{ for each }1\le i\le t\big\}|.
$$
Then $N>0$.
\end{proposition}
\begin{proof}
For $1 \le i \le t$, define $S_i(X)=(X-a_i)(b_iX-1)[(X-a_i)(b_iX-1)+4X] \in\FQ[X]$ and for $\es \ne I \subseteq[t]$, set $S_I=\prod_{i\in I}S_i$. Note that the leading coefficient of $S_i$ is a nonzero square $b_i^2$, then the leading coefficient of $S_I$ is a nonzero square in $\FQ^*$ for each $\es \ne I \subseteq[t]$. Set
$$
N'=\sum_{x\in\FQ\sm E}\ \prod_{i=1}^{t}\Big(1+\eta\big(S_i(x)\big)\Big).
$$
Then $N>0$ if and only if $N'>0$. We have
\begin{equation}
\label{eqn-Nodd}
\begin{split}
N'&=\sum_{x\in\FQ\sm E} \prod_{i=1}^{t}\Big(1+\eta\big(S_i(x)\big)\Big) \\
   &=\sum_{x\in\FQ\sm E}(1+\sum_{\es \subsetneq I\subseteq[t]} \eta\big(\prod_{ i \in I} S_i(x)\big)) \\
   &=Q-|E|+\sum_{x\in\FQ\sm E} \sum_{\es \subsetneq I\subseteq[t]} \eta\big(S_I(x)\big) \\
   &=Q-|E|+\sum_{\es \subsetneq I\subseteq[t]} \sum_{x\in\FQ\sm E} \eta\big(S_I(x)\big)
\end{split}
\end{equation}
Next, we focus on the estimate of $\sum_{x\in\FQ\sm E} \eta\big(S_I(x)\big)$. Define
$$
\cS=\big\{ f\in\FQ[X] \mid f=g^2 \text{ for some } g \in\ol{\FQ}[X] \big\}.
$$

If $S_I \in \cS$, then $S_I=a^2g^2$ with $a \in \FQ^*$ and $g \in \ol{\FQ}[X]$ is monic. We claim that $g \in \FQ[x]$. Let $\sig \in \mathrm{Gal}(\ol{\FQ}/\FQ)$ be the Frobenius automorphism such that $\sig(u)=u^{Q}$ for each $u \in \ol{\FQ}$. The automorphism $\sig$ applies to polynomials in $\ol{\FQ}[x]$ by acting on their coefficients. Since $S_I \in \FQ[X]$ and $a \in \FQ$, then
$$
\sig(g)^{2}=\sig\big(g^{2}\big)=\sig\big(S_I/a^2\big)=S_I/a^2=g^{2},
$$
which implies $\big(\sig(g)-g\big)\big(\sig(g)+g\big)=0$ in the integral domain $\ol{\FQ}[X]$. Therefore, $\sig(g)=g$ or $\sig(g)=-g$. Since $g$ is monic and the characteristic of $\FQ$ is odd, then we must have $\sig(g)=g$, implying $g \in \FQ[x]$. Therefore, $\eta(S_I(x))=\eta(a^2g(x)^2)=\eta(g(x))^2 \in \{0,1\}$. Consequently,
\begin{equation}
\label{eqn-SinS}
\sum_{x\in\FQ\sm E} \eta\big(S_I(x)\big) \ge 0
\end{equation}

If $S_I \notin \cS$, then $S_I$ is a nonzero polynomial of degree $4|I|$, which is not the square of a polynomial in $\ol{\FQ}[X]$. By Lemma \ref{lem-Weilquad}, 
$$
\Big|\sum_{x\in\FQ}\eta\big(S_I(x)\big)\Big| \le(4|I|-1)\sqrt Q .
$$
Note that $|E| \le 2t+1$, we have
\begin{equation}
\label{eqn-SnotinS}
\Big|\sum_{x\in\FQ \sm E}\eta\big(S_I(x)\big)\Big| \le(4|I|-1)\sqrt Q+2t+1
\end{equation}

Combining Equations \eqref{eqn-Nodd}, \eqref{eqn-SinS}, \eqref{eqn-SnotinS}, together with $\sum_{\es\ne I\subseteq[t]}|I|=\sum_{i=1}^{t}\binom{t}{i}i=t2^{t-1}$, $t \ge 2$, and $\sqrt{Q} \ge t2^{t+1}$, we have
\begin{align*}
N'&=Q-|E|+\sum_{\substack{\es \subsetneq I\subseteq[t] \\ S_I \in \cS} } \sum_{x\in\FQ\sm E} \eta\big(S_I(x)\big)+\sum_{\substack{\es \subsetneq I\subseteq[t] \\ S_I \not\in \cS} } \sum_{x\in\FQ\sm E} \eta\big(S_I(x)\big) \\
   &\ge Q-2t-1-\sum_{\es \subsetneq I\subseteq[t] } \big((4|I|-1)\sqrt Q+2t+1\big) \\
   &= Q-2t-1-\Big(4\,t2^{t-1}-(2^{t}-1)\Big)\sqrt Q-(2t+1)(2^{t}-1) \\
   &=Q-t2^{t+1}\sqrt Q+(2^{t}-1)\sqrt Q-2^t(2t+1) \\
   &=\sqrt Q\Big(\sqrt Q-t2^{t+1}\Big)+(2^{t}-1)\Big(\sqrt Q-2t-1\Big)-(2t+1) \\
   &\ge (2^{t}-1)\Big(2^{t+1}t-2t-1\Big)-(2t+1) \\
   &>0
\end{align*}
Consequently, $N'>0$ and thus, $N>0$. 
\end{proof}

\section*{AI Usage Disclosure}

During the preparation of this work, the authors used ChatGPT for the following purposes:
\begin{enumerate}[label=(\arabic*)]
\item surveying the recent literature on generalized covering radii and locating relevant references;
%\item verifying the specialization of \cite[Theorem~1.1]{CP06} to the form stated in Lemma~\ref{lem-Weilrational}, which led to the correction of an error in that specialization in an earlier version of this manuscript;
\item identifying and repairing a gap in the proof of Theorem~\ref{thm-asymub} in an earlier version of this manuscript;
\item proofreading.
\end{enumerate}
All definitions, statements, proofs and computations in this paper have been checked in detail by the authors, who take full responsibility for the content of this work.

\end{document}